\documentclass{IEEEtran}
\usepackage{cite}
\usepackage{amsmath,amssymb,amsfonts}
\usepackage{algorithmic}
\usepackage{textcomp}
\usepackage{graphicx}          % Include this line if your 
\usepackage{epsfig}    % or this line, depending on which
\usepackage{amssymb}
\usepackage{amsmath}
\usepackage{mathtools}
\usepackage{bm}
\usepackage{stfloats}
\usepackage{graphicx}
\usepackage[ruled]{algorithm}
\usepackage{epstopdf}
\usepackage{textcomp}
\usepackage{subcaption}
\usepackage{enumerate}
\usepackage{cite}
\usepackage{subeqnarray}
\usepackage{multirow}

\usepackage{xcolor,calc}

\newcommand{\qedsymbol}{\hfill$\blacksquare$}
\newtheorem{standing}{Standing Assumption}
\newtheorem{lemma}{\textit{Lemma}}
\newtheorem{remark}{Remark}
\newtheorem{prop}{Proposition}
\newtheorem{thm}{Theorem}
\newenvironment{proof}{\noindent\textit{Proof.}\ }{\hfill$\square$}
\def\BibTeX{{\rm B\kern-.05em{\sc i\kern-.025em b}\kern-.08em
    T\kern-.1667em\lower.7ex\hbox{E}\kern-.125emX}}
\begin{document}
\title{Neural network-based identification of state-space switching nonlinear systems}
\author{Yanxin Zhang, Chengpu Yu and Filippo Fabiani
\thanks{This work was supported by the National Natural Science Foundation of China (Grant No. 62088101, 62473046), Chongqing Natural Science Foundation CSTB2023NSCQ-JQX0018, and Beijing Natural Science Foundation L221005. (Corresponding author: Chengpu Yu).}

\thanks{Yanxin Zhang is with School of Automation, Beijing Institute of Technology, Beijing 100081, PR China (e-mail: zhangyanxin@bit.edu.cn).}
\thanks{Chengpu Yu is with School of Automation, Beijing Institute of Technology, Beijing 100081, PR China (e-mail: yuchengpu@bit.edu.cn).}
\thanks{Filippo Fabiani is with IMT School for Advanced Studies Lucca, Piazza San Francesco 19, 55100, Lucca, Italy (e-mail: filippo.fabiani@imtlucca.it).}
}

\maketitle

\begin{abstract}
We design specific neural networks (NNs) for the identification of switching nonlinear systems in the state-space form, which explicitly model the switching behavior and address the inherent coupling between system parameters and switching modes. This coupling is specifically addressed by leveraging the expectation-maximization (EM) framework. In particular, our technique will combine a moving window approach in the E-step to efficiently estimate the switching sequence, together with an extended Kalman filter (EKF) in the M-step to train the NNs with a quadratic convergence rate. Extensive numerical simulations, involving both academic examples and a battery charge management system case study, illustrate that our technique outperforms available ones in terms of parameter estimation accuracy, model fitting, and switching sequence identification.
\end{abstract}

\begin{IEEEkeywords}
Switching systems, Neural network, Expectation maximization, Extend Kalman filter, nonlinear system identification.
\end{IEEEkeywords}

\section{Introduction}

%在许多实际工程实践的例子中，动态系统无处不存在着切换这种现象。例如：。为了对这类动态系统进行分析与控制，研究系统内部结构以及对其建模变得尤为重要。这类系统的建模往往包含着对多个子系统参数的辨识，以及对每个时刻系统激活模态的估计。
In several engineering applications, such as speech recognition \cite{Schuller2008}, financial \cite{Timmermann,Guidolin,Tan2023}, and robotic systems \cite{Carloni2007,Schlegl2003}, the occurrence of mode switching is a pervasive feature characterizing dynamical systems. To effectively analyze and control these systems, it is then vital to investigate their internal structures and develop models that accurately represent their behavior. The process of modeling such systems, however, often entails the challenging task of parameter identification across various subsystems, as well as the estimation of the system's operational modes at any given time instant \cite{Piga2016b,Chan2008}.

%不幸运的是，系统的激活模态往往和子系统参数耦合在一起，这导致准确的辨识每个时刻激活的子系统和估计在什么时刻发生切换行为是困难的。因此，目前的方法都是将切换模态看作是一个隐含的离散状态，通过对时间分割进行聚类，得到属于一个子系统时间划分的集合。 然后再使用观测的输入输出数据对子系统的参数进行辨识。
Remarkably, the active modes of a switching system are often coupled with the parameters of its subsystems, a fact that complicates the identification of the active subsystem at each time instant and the estimation of the switching behaviors. Current methodologies treat the switching mode sequence as implicit discrete states, clustering the temporal segments to obtain a collection of temporal partitions belonging to a subsystem \cite{Bako2011,Ferrari2003}. Subsequently, the parameters of the subsystems are identified using observed input-output data \cite{Lauer2011,Ohlsson2010}.

\subsection{Related work}
It is well-known that the identification of switching systems amounts to an NP-hard problem in general \cite{Roll2004,Lauer2015}. Nevertheless, several methods have been proposed to address the typical identification issues for switching systems \cite{Porreca2009,Ripaccioli2009,Vidal2002,Xu2009}. Specifically, in \cite{survey} the authors survey the methodologies for identifying switching systems, also in the form of piecewise-affine (PWA) models, that have been developed during the last decade. More recently, some alternating identification methods have been proposed for jump models \cite{Vidal2002,Bemporad2018,Piga2020}. These approaches leverage the Bayesian theory to maximize the posterior probability function, thereby obtaining estimates of the switching sequence and the parameters of the subsystems. Available works on linear parameter varying (LPV) systems identification \cite{Toth2024,Piga2015,Golabi2017} employ a data-driven approach to model the LPV system with the influence of noise.  All the abovementioned works, however, focus on switching linear or PWA systems, while to the best of our knowledge the literature on switching nonlinear systems identification, especially in state-space form, is relatively poor.

Neural networks (NNs) serve as an effective means for modeling nonlinear systems and have been widely applied across various domains. In \cite{Batruni1991}, a multi-layer NN was applied to model PWA systems, while in \cite{Fabiani2025} it has been shown that the same class of models can be identified by combining an OptNet layer \cite{Amos2017} in cascade with a buffer one. In both approaches, the resulting NN could then be trained using back-propagation algorithms. For systems with (unknown) hidden states, some useful methods are introduced for learning the nonlinear state-space models in \cite{Prasad2003,Bemporad2021}. However, these methods rely on gradient-based algorithms, which are renowned for their slow convergence rate. Inspired by the fresh look of extended Kalman filter (EKF) in \cite{Humpherys2012}, a NN training algorithm based on the EKF is proposed in \cite{Bemporad2023}. However, these methods model the system using a NN as a black box, which can result in the loss of crucial internal information. For instance, when the system is composed of multiple subsystems, and only one of them is active at each time instant, the aforementioned approaches fail to capture the switching behavior of the system. Consequently, it is key to develop an algorithmic framework capable of modeling switching nonlinear systems without losing the information on their switching behavior.

\subsection{Summary of contribution and paper organization}
We devise a NN-based method for the identification of state-space switching nonlinear systems. Our method is developed in the expectation-maximization (EM) framework \cite{Mark2022}, thus consisting of two parts. In the E-step, the switching sequence is estimated by using a moving window approach. In the M-step, an EKF is used to obtain the estimation of the parameters of each subsystem. 

Our main contribution can be summarized as follows:
\begin{enumerate}
	\item We design a NN-based model able to represent nonlinear switching systems in the state-space form.
	\item The Markov property possessed by the system usually makes the switching sequence estimation computationally intensive \cite{Mark2022}. We adopt a moving window approach to drastically reduce such a computational burden.
	\item Given the time-consuming training of standard NNs and their lack of robustness to noise \cite{mlp,Anna2018,Jiao}, we develop an EKF-based NN training technique featuring a quadratic convergence rate.
\end{enumerate}
The performance of our method is finally tested through extensive numerical experiments on both academic examples and a real-world battery charge management system case study, which also aim at evaluating the efficiency of our EKF-based training scheme.

%\subsection{Paper organization}

The rest of the paper is organized as follows: in \S \ref{sec2} we describe the system and formalize the problem considered. In \S \ref{sec3} we discuss the identification method based on the EM framework. Specifically, \S \ref{sec3.1} gives the maximization step of the parameters in each subsystems by using the EKF-based method, while \S \ref{sec3.2} describes the expectation step based on a moving window approach to obtain the estimation of the switching sequence. The convergence analysis is then given in \S \ref{sec3.3}. Numerical simulations are discussed in \S \ref{sec4}. Finally, the conclusion and future work are given in \S \ref{sec5}. 
The proofs of the main technical results derived in the paper are deferred to Appendix \ref{sec:proofs}.

%\section{Notations}\label{sec2}
\subsection*{Notation}
$\mathbb{N}$, $\mathbb{Z}$, and $\mathbb{R}$ denote the set of natural, integer and real numbers, respectively. We indicate the extended real numbers as $\bar{\mathbb R}\coloneqq\mathbb{R}\cup\{+\infty\}$.
Given a matrix $X$, $\Vert X\Vert$ denotes its spectral norm, $\textrm{tr}(X)$ the trace, and $\textrm{vec}(X)$ the column vectorization of $X$. The operator $\mathrm{diag}(\cdot)$ produces a diagonal matrix with entries as its arguments. $\mathbb{P}[\cdot]$ and $\mathbb{E}[\cdot]$ respectively denote a probability distribution and the related expected value. $\mathbb{P}_\theta[\cdot]$ and $\mathbb{E}_\theta[\cdot]$ respectively denote the probability and expectation under the parameter $\theta$. For a sequence $\{s_t\}_{t\in\mathbb{N}}$, $s_T=\mathcal{O}(T)$ indicates that $\lim\sup_{T\to\infty}s_T/T<\infty$. $\mathbb{I}(\cdot,\cdot)$ denotes the standard indicator function, i.e., $\mathbb{I}(s(i),\hat{s}(i))=1$ if $s(i)=\hat{s}(i)$, $0$ otherwise. $I$ identifies a standard identity matrix, $N(\mu,\sigma^2)$ denotes the normal distribution of a random variable with mean $\mu$ and standard deviation $\sigma$.  

%\section{System description and problem formulation}\label{sec3}
\section{Mathematical formulation}\label{sec2}
\subsection{Dynamical system description}
In this paper we will consider switching nonlinear systems in the following state-space representation: 
\begin{align}
	x(t+1)&=f_{s_t}(x(t),u(t),\theta_{x,s_t})+\zeta(t),\label{eq1.1}\\
	y(t)&=g_{s_t}(x(t),u(t),\theta_{y,s_t})+\xi(t)\label{eq1.4},
\end{align}
where $t\in\mathbb{Z}$ is the discrete-time index, $x\in\mathbb{R}^{n_x}$ represents the vector of state variables, $u\in\mathbb{R}^{n_u}$ is the control input, and $y\in \mathbb{R}^{n_y}$ denotes the system output. The scalar variable $s_t\in\{1,\ldots,K\}$ denotes the hidden active mode at time $t$, which selects some $f_{s_t}$ and $g_{s_t}$ within a collection of $K$ nonlinear submodels. Specifically, for each $s_t\in\{1,\ldots,K\}$ we have $f_{s_t}:\mathbb{R}^{n_x}\times\mathbb{R}^{n_u}\times\mathbb{R}^{n_{\theta_x}}\to\mathbb{R}^{n_x}$ and $g_{s_t}:\mathbb{R}^{n_x}\times\mathbb{R}^{n_u}\times\mathbb{R}^{n_{\theta_y}}\to\mathbb{R}^{n_y}$, where $\theta_{x,s_t}\in\mathbb{R}^{n_{\theta_x}}$ and $\theta_{y,s_t}\in\mathbb{R}^{n_{\theta_y}}$ denote the parameters characterizing $f_{s_t}$ and $g_{s_t}$ at time instant $t$, respectively. Finally, $\zeta(t)\in\mathbb{R}^{n_x}$ and $\xi(t)\in\mathbb{R}^{n_y}$ represent the process and the measurement noises, respectively.
% -- see Fig.~\ref{fig:snssm} for an illustration of the architecture underlying the switching nonlinear state-space model at hand.

%In the remainder, we will assume that both noise terms $\zeta(t)$ and $\xi(t)$ follow a Gaussian distribution with zero mean and finite variance $\Sigma_1(t)$ and $\Sigma_2(t)$, respectively. 
\begin{standing}
	For all $t\in\mathbb{Z}$, $\zeta(t)\sim N(0,\Sigma_1(t))$ and $\xi(t)\sim N(0,\Sigma_2(t))$.
\end{standing}

%\begin{figure}
%	\centering
%	\includegraphics[height=0.7\linewidth]{pic/SNSSM.png} 
%	\caption{The switching nonlinear model in state-space form. Here, $\pi_{s_t,s_{t-1}}$ denotes the transition probability from mode $s_{t-1}$ to mode $s_t$. }
%	\label{fig:snssm}
%\end{figure}

\subsection{A recurrent neural network model}
Recurrent neural networks (RNNs) are usually designed for processing sequential data, and are characterized by their directional cyclic structure that enables to retain and utilize past information, making them appropriate for modeling causal dynamical systems \cite{Pan2024}. 

In the considered framework, we then propose to model the nonlinear, yet unknown, functions in \eqref{eq1.1} by means of tailored RNNs. Specifically, the state functions $f(\cdot)$ can be modeled by the following specific RNN:
\begin{equation}\label{eq1.2}
	\begin{aligned}
		h_f^1&=W_f^1\begin{bmatrix}
			x(t)\\
			u(t)
		\end{bmatrix}+b_f^1,\\
		h_f^2&=W_f^2 f^1(h_f^1)+b_f^2,\\
		&\vdots\\
		x(t+1)&=W_f^{l_f}f^{l_f-1}(h_f^{l_f-1})+b_f^{l_f}.
	\end{aligned}
\end{equation}
Similarly, the output function $g(\cdot)$ turns into:
\begin{equation}\label{eq1.3}
	\begin{aligned}
		h_g^1&=W_f^1\begin{bmatrix}
			x(t)\\
			u(t)
		\end{bmatrix}+b_g^1,\\
		h_g^2&=W_g^2 g^1(h_g^1)+b_g^2,\\
		&\vdots\\
		h_g^{l_g}&=W_g^{l_g}g^{l_g-1}(h_g^{l_g-1})+b_g^{l_g},\\
		y(t)&=g^{l_{g}}(h_g^{l_g}),
	\end{aligned}
\end{equation}
where $l_f$, $l_g>0$ represent the number of hidden layers of the state and output functions, respectively, $h_f^i$, $i=1,\ldots,l_f$ and $h_g^i$, $i=1,\ldots,l_g$ denote the output of the $i$-th layer within the RNN, which is then used as input to the $(i+1)$-th layer, $f^i$, $i=1,\ldots,l_f$ and $g^i$, $i=1,\ldots,l_g$ are the corresponding activation functions, such as $\tanh(\cdot)$ or rectified linear unit (ReLU). Finally, $W^i_f$, $i=1,\cdots,l_f$ and $W^i_g$, $i=1,\ldots,l_g$ are the weight matrices for the corresponding layer with appropriate dimensions, $b_f^i$, $i=1,\ldots,l_f$ and $b_g^i$, $i=1,\ldots,l_g$ are the associated vector of bias terms.

\begin{lemma}\textup{(Universal approximation theorem \cite{Hornik1989,Hornik1991})}\label{lemma1}
	The standard multilayer feedforward network with as few as a single hidden layer, and arbitrary bounded and nonconstant activation functions is a universal approximator with respect (w.r.t.) to any given continuous function.
\end{lemma}

\begin{remark}
	From Lemma \ref{lemma1}, a standard multilayer feedforward network can universally approximate any function. Given that RNN activation functions are a set of simple functions, any complex continuous nonlinear function $f(\cdot)$ can be uniformly approximated by polynomial functions. Thus, the RNN with structure in \eqref{eq1.2} can approximate any continuous function through its activation functions.
\end{remark}

Let us then collect the training parameters as $\theta_f\coloneqq\{(W_f^i,b_f^i)\}_{i=1}^{l_f}$, and $\theta_g\coloneqq\{(W_g^i,b_g^i)\}_{i=1}^{l_g}$. According to Lemma~\ref{lemma1}, the RNNs in \eqref{eq1.2}--\eqref{eq1.3} with parameters $\theta_f$ and $\theta_g$ can be used to describe the nonlinear functions in \eqref{eq1.1} and \eqref{eq1.4}. To simplify notation, we let $\mathcal{N}_x$ and $\mathcal{N}_y$ denote the RNNs with structure in \eqref{eq1.2} and \eqref{eq1.3}, respectively. 

By making use of the introduced RNNs, we can then rewrite the system in \eqref{eq1.1} and \eqref{eq1.4} as follows
\begin{equation}\label{eq1.5}
	\begin{aligned}
		x(t+1)&=\mathcal{N}_{x,s_t}(x(t),u(t),\theta_{f,s_t})+\zeta(t),\\
		y(t)&=\mathcal{N}_{y,s_t}(x(t),u(t),\theta_{g,s_t})+\xi(t),
	\end{aligned}
\end{equation}
where $\mathcal{N}_{x,s_t}$ and $\mathcal{N}_{y,s_t}$ are two groups of RNNs, each of them containing $K$-RNNs. Accordingly, $\theta_{f,s_t}$ and $\theta_{g,s_t}$ are the associated NN parameters. 

\subsection{Problem description}
Assume that we are able to collect a dataset consisting of $T$ samples of the system input and output (not necessarily a $T$-long trajectory), stacked together as $\bm{y}=\bm{y}_{1:T}\coloneqq\{y(1),\cdots,y(T)\}$ and $\bm{u}=\bm{u}_{1:T}\coloneqq\{u(1),\cdots,u(T)\}$. To establish our main technical results on the identification of the switching nonlinear state-space system in \eqref{eq1.1}--\eqref{eq1.4}, we make use of the following assumptions:
\begin{standing}\label{assum:1}
	For any $T>0$, the switching sequence $\bm{S}\coloneqq\{s_1,\cdots,s_T\}$, the system parameters $\Theta\coloneqq\{\theta_{f,1},\cdots,\theta_{f,K},\theta_{g,1},\cdots,\theta_{g,K}\}$, and the system inputs $\bm{u}$ are all independent among them, i.e.,
	\begin{align*}
		\mathbb P[\bm{S}\vert\Theta,\bm{u}]=\mathbb P[\bm{S}], \text{and } \,
		\mathbb P[\Theta\vert\bm{S},\bm{u}]=\mathbb P[\Theta].
	\end{align*}
\end{standing}
\begin{standing}\label{assum:2}
	The switching sequence satisfies the Markov property, i.e., for any $t\in\{1,\ldots,T\}$,
	\begin{align*}
		\mathbb P[s_t\vert s_{t-1},\cdots,s_1]=\mathbb P[s_t\vert s_{t-1}]=\pi_{s_{t},s_{t-1}}.
	\end{align*}
\end{standing}

Note that the aforementioned assumptions have simplified the complexity of the system, while possibly imposing restrictions. For instance, Standing Assumption~\ref{assum:1} implies that there is no correlation between the system parameters and the inputs. Similarly, the Markov property (i.e., Standing Assumption~\ref{assum:2}) assumes that the switching sequence depends solely on the previous state, without considering the influence of a longer time span. Similar conditions, however, have already been postulated in the cognate literature -- see, e.g., \cite{Bemporad2018,Piga2020}.
% utilization of these assumptions can be referenced in existing research within the field. For example, \cite{Bemporad2018,Piga2020} also employs similar assumptions to analyze different types of dynamical systems.

In view of Standing Assumption~\ref{assum:2}, the switching probability can then be characterized using a mode transition matrix $\Pi\in\mathbb{R}^{K \times K}$, which can be defined as follows:
\begin{align*}
	\Pi=\begin{bmatrix}
		\pi_{1,1}&\pi_{1,2}&\cdots&\pi_{1,K}\\
		\pi_{2,1}&\pi_{2,2}&\cdots&\pi_{2,K}\\
		\vdots&\vdots&\ddots&\vdots\\
		\pi_{K,1}&\pi_{K,2}&\cdots&\pi_{K,K}
	\end{bmatrix},
\end{align*}
and amounts to a row stochastic matrix since it satisfies
\begin{align}\label{eq1.29}
	\sum_{i=1}^{K}\pi_{i,j}=1,\text{ for all } j=1,\cdots,K.
\end{align}
By making use of the available data samples $\bm{u}$ and $\bm{y}$ collected from \eqref{eq1.1}--\eqref{eq1.4}, our goal is thus to determine $\mathcal N_x$ and $\mathcal N_y$ in \eqref{eq1.5} with a known number of system modes $K$.

To this end, given the model parameters $\Theta\in\mathbb{R}^d$, a suitable cost function to be minimized for the identification of the switching nonlinear systems consists of three parts: a pure loss function $\ell:\mathbb{R}^{n_u}\times\mathbb{R}^{n_y}\times\mathbb{R}^d\times\{1,\ldots,K\}^T\to\bar{\mathbb R}$, a regularization term for the parameters $r:\mathbb{R}^d\to\bar{\mathbb R}$, and a loss involving the mode sequence $\mathcal{L}:\{1,\ldots,K\}^T\to\bar{\mathbb R}$:
\begin{align}
	J(\bm{y},\bm{u},\Theta,\bm{S})\!=\!\ell(\bm{y},\bm{u},\Theta,\bm{S})+r(\Theta)+\mathcal{L}(\bm{S})\label{eq1.7}
\end{align}
Under Standing Assumption~\ref{assum:2}, from \cite{Bemporad2018} we note that $\mathcal{L}$ reads as: $\mathcal{L}(\bm{S})=\log\pi_{s_0}+\sum_{i=1}^{T}\log\pi_{s_t,s_{t-1}}$.

We then aim at determining the weights of the RNNs in \eqref{eq1.2}--\eqref{eq1.3} by minimizing $J(\bm{y},\bm{u},\Theta,\bm{S})$ w.r.t. $\Theta$, transition matrix $\Pi$ and switching sequence $\bm{S}$. Specifically, we will make use of the EM algorithm, an iterative method yielding an estimate of the underlying decision variables at each iteration, as detailed in the following section.

\section{The EM framework for the identification of switching nonlinear systems}\label{sec3}
By making use of the training data $\bm{u}$, $\bm{y}$ and the initial state $x(0)$, we aim at determining the parameters $\Theta$ of the RNNs $\mathcal N_x$ and $\mathcal N_y$ in \eqref{eq1.5}. The EM algorithm can iteratively obtain an estimate of $\Theta$ when the original system includes unobservable hidden variables, such as the switching sequence $\bm{S}$. 

Let us then denote the parameter estimate at the $k$-th iteration as $\Theta^k$. The likelihood function associated to the data collected over $T$ of $\bm{y}$, $\bm{x}=\bm{x}_{1:T}\coloneqq\{x(1),\ldots,x(T)\}$, $\bm{S}$, and $\Theta$, can be expressed as:
\begin{align}\label{eq1.28}
	\log\mathbb{P}[\bm{y},\bm{x},\bm{S},\Theta]=\log\mathbb{P}[\bm{y}]+\log\mathbb{P}[\bm{x},\bm{S},\Theta\vert\bm{y}].
\end{align}
Given some $\Theta^k$,  the conditional expectation of $ \mathbb P_{\Theta^k}[\bm{x},\bm{S}\vert\bm{y}]$ is abbreviated as
\begin{align*}
		\mathbb{E}_{\Theta_k}[\cdot]=\int\sum_{\bm{S}}(\cdot)\mathbb P_{\Theta^k}[\bm{x},\bm{S}\vert\bm{y}]d(\bm{x})
\end{align*}
Then, take the expectation operator $\mathbb{E}_{\Theta_k}[\cdot]$ on both sides of \eqref{eq1.28}:
%we hence take the expectation on both sides of \eqref{eq1.28} w.r.t. to the conditional distribution $\mathbb{P}_{\Theta^k}[\cdot]$, which in short reads as $\mathbb{E}_{\Theta^k}[\cdot]$.
% Then, from the Bayes' rule and Markov property it can be inferred that:
\begin{align}\label{eq1.34}
	\mathbb{E}_{\Theta^k}[\log\mathbb{P}[\bm{y},\bm{x},\bm{S},\Theta]]=Q_1+Q_2+Q_3,
\end{align}
where
\begin{align*}
	Q_1=&\sum_{i=1}^{T}\int\sum_{s_i}\log\mathbb{P}[y(i)\vert x(i),u(i),\theta_{g,s_i}]\\
	&\hspace{3.7cm}\mathbb{P}_{\Theta^k}[x(i),s(i)\vert \bm{y}]d(x(i))\\
	&+\sum_{i=1}^{T-1}\int\int\sum_{s_i}\log\mathbb{P}[x(i+1)\vert x(i),u(i),\theta_{f,s_i}]\\
	&\hspace{1.2cm}\mathbb{P}_{\Theta^k}[x(i+1),x(i),s(i)\vert \bm{y}]d(x(i+1))d(x(i)),\\
	Q_2=&\sum_{i=1}^{K}\log\mathbb{P}[\theta_{f,s_i}]\mathbb{P}[\theta_{g,s_i}],\\
	Q_3=&\sum_{j=1}^{K}\sum_{l=1}^{K}\mathbb{E}_{\Theta^k}\big[\sum_{i=1}^{T}\mathbb{I}(j,s_i)\mathbb{I}(l,s_{i-1})\big]\log\pi_{j,l}+\log\pi_{s_0}.
\end{align*}
With $\bm{\theta}_f\coloneqq\{\theta_{f,1},\cdots,\theta_{f,K}\}$ and $\bm{\theta}_g\coloneqq\{\theta_{g,1},\cdots,\theta_{g,K}\}$, we show next that the cost in \eqref{eq1.7} can be obtained in the maximum likelihood estimation framework

\begin{prop}\label{prop1}
	Minimizing the cost function $J(\bm{y},\bm{u},\Theta,\bm{S})$ in \eqref{eq1.7} w.r.t. $\Theta$, where
	\begin{subequations}\label{subeq1.6}
		\begin{align}
			\ell(&\bm{y},\bm{u},\Theta,\bm{S})=-\sum_{t=1}^{T}\log\mathbb{P}[y(t)\vert x(t),u(t),\theta_{g,s_t}]\nonumber\\
			&\hspace{1.5cm}-\sum_{t=1}^{T-1}\log\mathbb{P}[x(t+1)\vert x(t), u(t),\theta_{f,s_{t}}],\label{subeq1.61}\\
			r&(\Theta)=-\sum_{i=1}^{K}(\log\mathbb{P}[\theta_{f,i}]+\log\mathbb{P}[\theta_{g,i}]),\\
			%			r&(\theta_{f,i})=-\log\mathbb{P}[\theta_{f,i}],\\
			%			r&(\theta_{g,i})=-\log\mathbb{P}[\theta_{g,i}],\\
			\mathcal{L}&(\bm{S})=-\log\mathbb{P}[\bm{S}]=-\sum_{t=1}^{T}\log\pi_{s_t,s_{t-1}}-\log\pi_{s_0},\label{subeq1.62}
		\end{align}
	\end{subequations}
	is equivalent to maximizing the joint probability density function $\mathbb{P}[\bm{y},\bm{x},\bm{S},\Theta]$.
\end{prop}
The proof of the Proposition \ref{prop1} is shown in Appendix A.
\begin{remark}
	Proposition \ref{prop1} establishes a connection between the cost function and the joint probability density function (pdf). Specifically, it transforms a deterministic training problem into a probabilistic one. We will therefore be able to leverage tools from probability theory and statistics to address our challenging identification problem, as the fact that the pdf contains unknown states  $\bm{S}$ will be resolved iteratively by our EM-based algorithm.
\end{remark}

It thus follows from Proposition \ref{prop1} that, rather than minimizing $J(\bm{y},\bm{u},\Theta,\bm{S})$ with ingredients as in \eqref{subeq1.6}, we can equivalently maximize the following objective function:
\begin{align*}
	Q(\Theta,\Theta^k)=\mathbb{E}_{\Theta^k}[\log\mathbb{P}[\bm{y},\bm{x},\bm{S},\Theta]].
\end{align*}
Conceptually, the proposed EM-based methodology then consists in the two steps summarized next:
\begin{enumerate}
	\item \textbf{Expectation} (E-step): Calculate the optimal posterior estimate values of the implicit state (switching sequence $\bm{S}$) using the parameter estimates obtained in the M-step at the previous iteration, and compute the expectation to obtain $Q(\Theta,\Theta^k)$;
	\item \textbf{Maximization} (M-step): The objective function $Q(\Theta,\Theta^k)$ is maximized w.r.t. $\Theta$ and $\Pi$ by using an EKF-based method. Then, the estimate of the parameter $\Theta$ is updated to obtain $\Theta^{k+1}$.
\end{enumerate}

Next, each one of the steps above will be detailed
starting from the M-step in \S \ref{sec3.1}, which will introduce the ingredients required for the E-step in \S \ref{sec3.2}.

\subsection{The maximization step}\label{sec3.1}
Given its expression, maximizing $Q(\Theta,\Theta^k)$ yields the optimal parameter estimates and transition matrix $\Pi$ together. By starting with $\Pi$, it is clear that a maximizer to $Q(\Theta,\Theta^k)$ w.r.t. $\Pi$ can be equivalently found as a maximizer to $Q_3$ subject to \eqref{eq1.29}. Then, the entries of the transition matrix can be calculated, for all $l=1,\ldots,K$, as: 
\begin{align}\label{eq1.30}
	\pi_{j,l}&=\frac{\mathbb{E}_{\Theta^k}\big[\sum_{i=1}^{T}\mathbb{I}(j,s_i)\mathbb{I}(l,s_{i-1})\big]}{\sum_{j=1}^{K}\mathbb{E}_{\Theta^k}\big[\sum_{i=1}^{T}\mathbb{I}(j,s_i)\mathbb{I}(l,s_{i-1})\big]}
\end{align} 
Then, the optimal parameter estimates can be given by:
\begin{align}\label{eq1.8}
	\Theta^{k+1}=\arg\max_{\Theta}~Q(\Theta,\Theta^k).
\end{align}
To this end, a common approach is the maximum likelihood (ML) technique \cite{Umenberger2018} that requires one to calculate the extreme points of the parameter (correspond to the optimal estimates). Therefore, the gradient descent (GD) method can be applied to solve the maximization problem \eqref{eq1.8} iteratively with update
\begin{align*}
	\Theta^{n+1}=\Theta^n-\alpha_n (\partial Q(\Theta,\Theta^n)/\partial \Theta),
\end{align*}
where $n=1,\ldots,N$ is the training epoch index, with associated learning rate $\alpha_n$. The gradients of $Q(\Theta,\Theta^k)$ w.r.t. all parameters in the RNNs \eqref{eq1.2}--\eqref{eq1.3} can then be calculated by applying the chain rule. However, the GD method has few disadvantages, such as slow convergence rate when the gradient is close to zero, and over-fitting, especially in case of noise-corrupted data. Although there are currently many well-known algorithms and NN structures proposed to address these problems, such as SGD, Adam, BiGRU; however, the objective function contains unknown state parameters $\bm{x}$ and $\bm{S}$, which poses several challenges to their differentiation.

Inspired by \cite{Bemporad2023}, we propose instead to employ an EKF to recursively update the parameter $\Theta$ and state trajectory $\bm{x}$. For computational convenience, the parameter will be vectorized and the dynamics in \eqref{eq1.5} rewritten as follows:
\begin{align}\label{eq1.9}
	x(t+1)&=\mathcal{N}_{x,s_t}(x(t),u(t),\theta_{f,s_t})+\zeta(t),\nonumber\\
	y(t)&=\mathcal{N}_{y,s_t}(x(t),u(t),\theta_{g,s_t})+\xi(t),\\
	\vartheta(t+1)&=\vartheta(t)+\eta(t)\nonumber,
\end{align}
with $\vartheta\coloneqq[\vartheta_{f,s_t}^\top~\vartheta_{g,s_t}^\top]^\top$, and $\vartheta_{f,s_t}\coloneqq\textrm{vec}(\theta_{f,s_t})$, $\vartheta_{g,s_t}\coloneqq\textrm{vec}(\theta_{g,s_t})$. In \eqref{eq1.9} we have implicitly assumed that the parameters vector dynamics $\vartheta$ is affected by Gaussian white noise $\eta(t)$ with zero mean and variance $\Sigma_\vartheta(t)$

We discuss next the update process of the EKF at time $t$. First, the nonlinear functions $\mathcal{N}_{x,s_t}$ and $\mathcal{N}_{y,s_t}$ need to be expanded in a Taylor series to perform the first-order linearization of the system \eqref{eq1.9}. The corresponding Jacobian matrices read as follows:
\begin{align}
	F(t)&=\left.\frac{\partial\mathcal{N}_{x,s_t}}{\partial\begin{bmatrix}
			x\\
			\vartheta
	\end{bmatrix}}\right|_{\hat{x}(t),\hat{\vartheta}(t),u(t)}=\begin{bmatrix}
		\frac{\partial\mathcal{N}_{x,s_t}}{\partial x}&\frac{\partial\mathcal{N}_{x,s_t}}{\partial \vartheta_{f,s_t}}&0\\
		0&I&0\\
		0&0&I
	\end{bmatrix},\label{eq1.10}\\
	H(t)&=\left.\frac{\partial\mathcal{N}_{y,s_t}}{\partial\begin{bmatrix}
			x\\
			\vartheta
	\end{bmatrix}}\right|_{\hat{x}(t),\hat{\vartheta}(t),u(t)}=\begin{bmatrix}
		\frac{\partial\mathcal{N}_{y,s_t}}{\partial x}&0&\frac{\partial\mathcal{N}_{y,s_t}}{\partial \vartheta_{g,s_t}}
	\end{bmatrix}.\label{eq1.11}
\end{align}
Then, the prior estimates of $\vartheta(t+1)$ and state $x(t+1)$, denoted as $\hat{\vartheta}^-(t+1)$ and $\hat{x}^-(t+1)$, can be computed by exploiting the forward propagation process of the RNNs:
\begin{align}\label{eq1.12}
	\begin{bmatrix}
		\hat{x}^-(t+1)\\
		\hat{\vartheta}^-(t+1)
	\end{bmatrix}=\begin{bmatrix}
		\mathcal{N}_{x,s_t}(\hat{x}(t),u(t),\hat{\vartheta}_{f,s_t}(t))\\
		\hat{\vartheta}(t)
	\end{bmatrix}.
\end{align}
With this regard, the measurement equation can then be used to update and correct the prior estimates of the parameter $\hat{\vartheta}^-(t+1)$ and state $\hat{x}^-(t+1)$. In general, they are referred to as the posterior estimates of the parameter and state, denoted as $\hat{\vartheta}(t+1)$ and $\hat{x}(t+1)$:
\begin{align}\label{eq1.13}
	\begin{bmatrix}
		\hat{x}(t+1)\\
		\hat{\vartheta}(t+1)
	\end{bmatrix}=\begin{bmatrix}
		\hat{x}^-(t+1)\\
		\hat{\vartheta}^-(t+1)
	\end{bmatrix}+\varGamma(t)e(t),
\end{align}
where $\varGamma(t)$ is the Kalman gain, $e(t)$ is the posterior estimation error that can be calculated by the forward propagation of $\mathcal{N}_{y,s_t}$ and the true output, i.e.,
\begin{align}\label{eq1.14}
	e(t)=y(t)-\mathcal{N}_{y,s_t}(\hat{x}^-(t),u(t),\hat{\vartheta}^-_{g,s_t}(t)).
\end{align}
Thus, the optimal Kalman gain $\varGamma(t)$ has to be found so as to minimizes $e(t)$, which is equivalent to minimizing the covariance matrix of the posterior estimation error. Let us then define the prior and posterior error covariance matrices as:
\begin{align*}
	P^-(t)&\!\coloneqq\!\mathbb{E}\left[\left(\begin{bmatrix}
		x(t)\\
		\vartheta(t)
	\end{bmatrix}\!-\!\begin{bmatrix}
		\hat{x}^-(t)\\
		\hat{\vartheta}^-(t)
	\end{bmatrix}\right)\left(\begin{bmatrix}
		x(t)\\
		\vartheta(t)
	\end{bmatrix}\!-\!\begin{bmatrix}
		\hat{x}^-(t)\\
		\hat{\vartheta}^-(t)
	\end{bmatrix}\right)^\top\right],\\
	P(t)&\!\coloneqq\!\mathbb{E}\left[\left(\begin{bmatrix}
		x(t)\\
		\vartheta(t)
	\end{bmatrix}\!-\!\begin{bmatrix}
		\hat{x}(t)\\
		\hat{\vartheta}(t)
	\end{bmatrix}\right)\left(\begin{bmatrix}
		x(t)\\
		\vartheta(t)
	\end{bmatrix}\!-\!\begin{bmatrix}
		\hat{x}(t)\\
		\hat{\vartheta}(t)
	\end{bmatrix}\right)^\top\right].
\end{align*}
Then, we have:
\begin{align}\label{eq1.15}
	P^-(t)&\!=\!\mathbb{E}\left[\left(F(t)\left(\begin{bmatrix}
		x(t-1)\\
		\vartheta(t-1)
	\end{bmatrix}\!-\!\begin{bmatrix}
		\hat{x}(t-1)\\
		\hat{\vartheta}(t-1)
	\end{bmatrix}\right)\!+\!\begin{bmatrix}
		\zeta(t)\\
		\eta(t)
	\end{bmatrix}\right)\right.\nonumber\\
	&\left.\left(F(t)\left(\begin{bmatrix}
		x(t-1)\\
		\vartheta(t-1)
	\end{bmatrix}\!-\!\begin{bmatrix}
		\hat{x}(t-1)\\
		\hat{\vartheta}(t-1)
	\end{bmatrix}\right)\!+\!\begin{bmatrix}
		\zeta(t)\\
		\eta(t)
	\end{bmatrix}\right)^\top\right]\nonumber\\
	&=F(t)P(t-1)F^\top(t)+\begin{bmatrix}
		\Sigma_1(t)&0\\
		0&\Sigma_\vartheta(t)
	\end{bmatrix},
\end{align}
and
\begin{align}\label{eq1.16}
	P(t)&=\mathbb{E}\left[\left((I-\varGamma(t)H(t))(x(t)-\hat{x}^-(t))-\varGamma(t)\xi(t)\right)\right.\nonumber\\
	&\left.\left((I-\varGamma(t)H(t))(x(t)-\hat{x}^-(t))-\varGamma(t)\xi(t)\right)^\top\right]\nonumber\\
	&=P^-(t)-P^-(t)H^\top(t)\varGamma^\top(t)-\varGamma(t)H(t)P^-(t)\nonumber\\
	&+\varGamma(t)H(t)P^-(t)H^\top(t)\varGamma^\top(t)+\varGamma(t)\Sigma_2(t)\varGamma^\top(t).
\end{align}
Moreover, note that minimizing the covariance matrix $P(t)$ is equivalent to minimizing its trace. Thus, the optimal $\varGamma(t)$ assumes the following closed form expression:
\begin{align}\label{eq1.17}
	\varGamma(t)=P^-(t)H^\top(t)\left[H(t)P^-(t)H^\top(t)+\Sigma_2(t)\right]^{-1}.
\end{align} 
Finally, by substituting the obtained Kalman gain into \eqref{eq1.10}, the updated error covariance matrices can be calculated for the subsequent iteration as:
\begin{align}\label{eq1.18}
	P(t)=(I-\varGamma(t)H(t))P^-(t).
\end{align} 
The instructions in \eqref{eq1.10}-\eqref{eq1.18} then summarize the iterative training process based on EKF for each submodel. Before performing \eqref{eq1.10}-\eqref{eq1.18}, it is however key to cluster the training data by using the results provided in the E-step.  We will detail this process in \S \ref{sec3.2}. We now show that performing the aforementioned submodel EKF-based training steps, for each submodel, is equivalent to maximizing the objective function $Q(\Theta,\Theta^k)$:
\begin{thm}\label{th1}
	Performing the steps in \eqref{eq1.10}-\eqref{eq1.18}, separately for each submodel, is equivalent to maximizing the objective function $Q(\Theta,\Theta^k)$ w.r.t. $\Theta$.	
\end{thm}

The proof of Theorem~\ref{th1} is shown in Appendix \ref{sec:proofs}.

Theorem~\ref{th1} proves that using EKF to update parameters and state variables is equivalent to directly maximizing the objective function $Q(\Theta,\Theta^k)$, thus simplifying heavily the parameter estimation process. The proposed EKF-based training alleviates the computational difficulties that may arise when directly maximizing the objective function through a recursive procedure. 

\subsection{The expectation step}\label{sec3.2}
The maximization step described in \S \ref{sec3.1} requires an expression for the expectation of the switching sequence. Specifically, this shall be obtained based on the parameter $\Theta^k$ and the state variables obtained from the maximization step at the previous iteration.

%First, from the formula of total probability, we obtain:
%\begin{align*}
%	\mathbb{P}[\bm{y},\bm{x},\Theta]=\prod_{t=1}^{T}\prod_{i=1}^{K}\mathbb{P}[\bm{y},\bm{x},\Theta\vert s_t=i]\mathbb{P}[s_t=i].
%\end{align*}
Due to the Markov property of the switching sequence, which means that the hidden state $s_t$ at the current time is only related to the previous time, we have:
\begin{align}\label{eq1.19}
	\mathbb{P}[s_t=i]&=\sum_{j=1}^{K}\mathbb{P}[s_t=i,s_{t-1}=j]\nonumber\\
	&=\sum_{j=1}^{K}\mathbb{P}[s_t=i\vert s_{t-1}=j]\mathbb{P}[s_{t-1}=j].
\end{align}

%Therefore, the posterior probability of $\bm{y}$, $\bm{x}$, and $\Theta^k$ can be written as 
%\begin{align*}
%	\mathbb{P}[\bm{y},\bm{x},\Theta]&=\prod_{t=1}^{T}\prod_{i=1}^{K}\prod_{j=1}^{K}\mathbb{P}[\bm{y},\bm{x},\Theta\vert s_t=i]\\
%	&\hspace{3cm}\mathbb{P}[s_t=i\vert s_{t-1}=j]\\
%	&=\prod_{t=1}^{T}\prod_{i=1}^{K}\prod_{j=1}^{K}\mathbb{P}[\bm{y},\bm{x},\Theta\vert s_t=i]\pi_{i,j}.
%\end{align*}
Then, the posterior pdf of the switching sequence can be calculated by using the Bayes' rule, i.e., $\mathbb{P}[\bm{S}\vert \bm{y},\bm{x},\Theta]=\mathbb{P}[\bm{y},\bm{x},\Theta,\bm{S}]/\mathbb{P}[\bm{y},\bm{x},\Theta]$.
By maximizing the pdf of the switching sequence $\bm{S}$, we have:
\begin{align}\label{eq1.23}
	\hat{\bm{S}}&=\arg\max_{\bm{S}}\log\mathbb{P}[\bm{S}\vert \bm{y},\bm{x},\Theta]\nonumber\\
	&=\arg\max_{\bm{S}}\log\mathbb{P}[\bm{y},\bm{x},\Theta,\bm{S}]\nonumber\\
	&=\arg\max_{\bm{S}}\ell(\bm{y},\bm{x},\Theta,\bm{S})+r(\Theta)\nonumber\\
	&+\sum_{t=1}^{T}\sum_{i=1}^{K}\sum_{j=1}^{K}\mathbb{I}(s_t,i)\mathbb{I}(s_{t-1},j)\log\pi_{i,j}
\end{align}
Unfortunately, the calculation of \eqref{eq1.23} requires to try all possibilities of $\bm{S}$, which is however computationally challenging to obtain. To derive the estimate for the switching sequence, a number of $K^{T+1}$ calculations are thus necessary. To alleviate the resulting computational burden, we adopt a moving window approach for calculating these posterior probabilities. In words, this method maintains the estimated switching mode from the previous time instant, while calculating the pdf for the switching mode at the current time instant. Successively, the estimation of the current mode is obtained by maximizing the pdf of the switching mode at the current time instant.

We now provide a specific procedure starting from $t=0$. Specifically, we note that the loss of the switching sequence \eqref{subeq1.62} is only related to the initial mode. In fact, for all possible modes the optimal one can be chosen as: 
\begin{align}\label{eq1.24}
	\hat{s}_1&=\arg\max_{s_1}\mathbb{P}[s_1\vert y(1),x(1),u(1),\Theta]\nonumber\\
	&=\arg\min_{s_1}\ell(y(1),u(1),\theta_{f,s_1},\theta_{g,s_1})+r(\Theta)+\mathcal{L}(s_1).
\end{align}
Note that, when the initial mode is determined, the step-forward can then be computed for all possible switching sequences $\bm{S}_{2:T}$. Assume the length of the considered time window is $T_w$, and denote the time window as $\bm{W}_t=\{t, t+1,\ldots,t+T_w-1\}$. Then, the switching sequence $\bm{S}_{\bm{W}_t}$ can be computed by maximizing the local posterior pdf in \eqref{eq1.23}:
\begin{align}\label{eq1.25}
	\hat{\bm{S}}_{\bm{W}_t}&=\arg\max_{\bm{S}_{\bm{W}_t}}\mathbb{P}[\bm{S}_{\bm{W}_t}\vert \bm{y}_{\bm{W}_t},\bm{x}_{\bm{W}_t},\bm{u}_{\bm{W}_t},\Theta]\nonumber\\
	&\hspace{1.2cm}\textrm{s.t.}\quad s_{t-1}=\hat{s}_{t-1},\nonumber\\
	&=\arg\min_{\bm{S}_{\bm{W}_t}}\ell(\bm{y}_{\bm{W}_t},\bm{u}_{\bm{W}_t},\Theta)+r(\Theta)+\sum_{t\in \bm{W}_t}\log\pi_{s_t,s_{t-1}},\nonumber\\
	&\hspace{1.2cm}\textrm{s.t.}\quad s_{t-1}=\hat{s}_{t-1}.
\end{align}
The idea is thus to calculate an optimal solution to \eqref{eq1.25} for all the possible switching sequences $\bm{S}_{\bm{W}_t}$. Then, the optimal mode at time instant $t$ can be fixed as the first element of the $\bm{S}_{\bm{W}_t}$, i.e., $\hat{s}_t=\hat{\bm{S}}_{\bm{W}_t}(1)$. This process is then repeated up to $t=T-T_w+1$, and hence obtaining an estimate of switching sequence $\hat{\bm{S}}$.

Remarkably, with the proposed approach the posterior pdf is calculated $K^{T_w+1}$ times for each time instant $t=2,\cdots,T-T_w+1$. Thus, the computational complexity of the moving window approach is only $O((T-T_w)K^{T_w+2})$.
\begin{remark}
	The time window length $T_w\in\{1,\cdots,T-1\}$ offers a trade-off between the accuracy and computational complexity. When $T_w=T-1$, the moving window approach degenerates into considering all possible of switching sequences, and the computational complexity increases up to $O(K^{T+1})$.
\end{remark}
\begin{remark}
	The moving window method calculates the posterior probability at each time instant by only considering the possible modes at the current time and connecting them with the estimated mode from the previous time. This approach alleviates the combinatorial nature of enumerating all possible sequences, reducing the time complexity from exponential to linear order, and significantly improving the running speed of the algorithm, making it suitable for processing large-scale dataset. 
\end{remark}

\begin{algorithm}[t!]
	\caption{EM-based identification of switching nonlinear system}\label{alg:EM}
	\smallskip
	
	\textbf{Initialization:} Collect data $\bm y_{1:T}, \bm{u}_{1:T}$,  set $\Theta^0$, retrieve number of the modes $K$
	
	\smallskip
	
	\textbf{Iteration} $k\in\mathbb{Z}$\textbf{:}
	\smallskip
	\begin{enumerate}
		\item \textbf{E-step}: Compute $Q(\Theta,\Theta^k)$ using \eqref{eq1.34}, \eqref{eq1.23}, \eqref{eq1.24}, \eqref{eq1.25}
		\smallskip
		\item \textbf{M-step}: Maximize $Q(\Theta,\Theta^k)$ w.r.t. $\Theta$ and $\Pi$ by using \eqref{eq1.30} and the EKF-based procedure \eqref{eq1.10}-\eqref{eq1.18} 
	\end{enumerate}
\end{algorithm}
\subsection{Convergence analysis}\label{sec3.3}
The main steps of our method are summarized in Algorithm \ref{alg:EM}. We characterize next its convergence properties.
\begin{prop}\label{th3}
	Given a set of data $\bm{y}$, $\bm{u}$, and a number of modes $K$, let $\{\Theta^k\}_{k\in\mathbb{Z}}$ be the sequence generated by Algorithm \ref{alg:EM}. Then, the log likelihood function $\log\mathbb{P}_\Theta[\bm{y}]$ is non-decreasing, i.e., $\log\mathbb{P}_{\Theta^{k+1}}[\bm{y}]\geq\ln\mathbb{P}_{\Theta^{k}}[\bm{y}]$.
\end{prop}
\begin{proof}
	According to Theorem~\ref{th1}, performing the EKF steps \eqref{eq1.10}--\eqref{eq1.18} is equivalent to maximizing the objective function $Q(\Theta,\Theta^k)$, which yields at each iteration
	\begin{align}\label{eq1.26}
		Q(\Theta,\Theta^{k+1})\geq Q(\Theta,\Theta^k).
	\end{align}
	Then, we have:
	\begin{align*}
		\log\mathbb{P}_{\Theta}[\bm{y}]-&\log\mathbb{P}_{\Theta^{k}}[\bm{y}]\\
		&=Q(\Theta,\Theta^k)-Q(\Theta^k,\Theta^k)\\
		&+\int\log\mathbb{P}_{\Theta^k}[\bm{x},\bm{S}\vert\bm{y}]\log\mathbb{P}_{\Theta^k}[\bm{x},\bm{S}\vert\bm y]d(\bm{x})\\
		&-\int\log\mathbb{P}_{\Theta}[\bm{x},\bm{S}\vert\bm{y}]\log\mathbb{P}_{\Theta^k}[\bm{x},\bm{S}\vert\bm y]d(\bm{x})\\
		&=Q(\Theta,\Theta^k)-Q(\Theta^k,\Theta^k)\\
		&+\int\frac{\log\mathbb{P}_{\Theta^k}[\bm{x},\bm{S}\vert\bm y]}{\log\mathbb{P}_{\Theta}[\bm{x},\bm{S}\vert\bm{y}]}\log\mathbb{P}_{\Theta^k}[\bm{x},\bm{S}\vert\bm{y}]d(\bm{x})\\
		&=Q(\Theta,\Theta^k)-Q(\Theta^k,\Theta^k)\\
		&+DL(\log\mathbb{P}_{\Theta^k}[\bm{x},\bm{S}\vert\bm y]\Vert \log\mathbb{P}_{\Theta}[\bm{x},\bm{S}\vert\bm{y}]).
	\end{align*}
	where $DL(P\Vert Q)$ is the Kullback-Leibler divergence, which is guaranteed to be nonnegative \cite{Ji2022}. Therefore, applying \eqref{eq1.26} directly lead to a non-decreasing sequence for $\log\mathbb{P}_{\Theta}[\bm{y}]$, i.e., $\log\mathbb{P}_{\Theta^{k+1}}[\bm{y}]\geq\ln\mathbb{P}_{\Theta^{k}}[\bm{y}]$,
	completing the proof.
\end{proof}

Proposition \ref{th3} is a well-known result for EM method \cite{Mark2022} which says that the parameter estimates generated at each iteration of Algorithm \ref{alg:EM} are approximating the optimal value of the maximum likelihood estimate. By denoting with $\Theta^*$ as an optimal solution to the minimization problem of \eqref{eq1.7}, we can claim the following result:
\begin{thm}\label{th4}
	Denote the Hessian matrix of $Q(\Theta,\Theta^k)$ as $\mathcal{H}(\Theta,\Theta^k)$, assume its inverse being bounded, and $\Theta\to\mathcal{H}(\Theta,\Theta^k)$ Lipschitz continuous. Then, the sequence $\{\Theta^k\}_{k\in\mathbb{Z}}$ generated by Algorithm \ref{alg:EM} converges quadratically to some optimal solution $\Theta^*$.
	%	which means
	%	\begin{align*}
		%		\frac{\Vert \Theta^{k+1}-\Theta^*\Vert}{\Vert \Theta^k-\Theta^*\Vert^2}\leq \mathcal{V}
		%	\end{align*}
	%	where $\mathcal{V}$ is a positive real number.
\end{thm}
\begin{proof}
	According to the EKF-based training process for the parameters update, and \eqref{eq1.31}--\eqref{eq1.32}, we have:
	\begin{align}\label{eq1.33}
		&\vartheta(t+1)\nonumber\\
		&=\vartheta(t)+P^-(t)H^\top(t)\left[H(t)P^-(t)H^\top(t)+\Sigma_2(t)\right]^{-1}e(t)\nonumber\\
		&=\vartheta(t)\!+\!(P^-(t)^{-1}\!+\!H^\top(t)\Sigma_2^\top(t)H(t))^{-1}H^\top(t)\Sigma_2^{-1}(t)e(t)\nonumber\\
		&=\vartheta(t)+P(t)H^\top(t)\Sigma_2^{-1}(t)(y(t)-\mathcal{N}_{y,s_t}(\vartheta(t))),
	\end{align}
	where 
	\begin{align*}
		\frac{\partial Q(\vartheta(t),\Theta^k)}{\partial \vartheta(t)}&=-H^\top(t)\Sigma_2^{-1}(t)e(t),\\
		\mathcal{H}(\vartheta(t),\Theta^k)&=\frac{\partial^2 Q(\vartheta(t),\Theta^k)}{\partial \vartheta(t)^2}\\
		&=P^-(t)^{-1}+H^\top(t)\Sigma_2^\top(t)H(t).
	\end{align*}
	Let us define $\mathcal{F}(a)=H^\top(t)\Sigma_2^{-1}(t)(y(t)-\mathcal{N}_{y,s_t}(\vartheta(t)+a(\Theta^*-\vartheta(t))))$, and its derivative as $\mathcal{F}^{'}(a)=\mathcal{H}(\vartheta(t)+a(\Theta^*-\vartheta(t)),\Theta^k)(\Theta^*-\vartheta(t))$. According to Newton-Leibniz formula, we readily obtain:
	\begin{align*}
		&-\frac{\partial Q(\vartheta(t),\Theta^k)}{\partial \vartheta(t)}\\
		&=H^\top(t)\Sigma_2^{-1}(t)(y(t)-\mathcal{N}_{y,s_t}(\Theta^*))\\
		&\hspace{3cm}-H^\top(t)\Sigma_2^{-1}(t)(y(t)-\mathcal{N}_{y,s_t}(\vartheta(t)))\\
		&=\mathcal{F}(1)-\mathcal{F}(0)\\
		&=\int_{0}^{1}\mathcal{F}^{'}(a)d(a)\\
		&=\int_{0}^{1}\mathcal{H}(\vartheta(t)+a(\Theta^*-\vartheta(t)),\Theta^k)(\Theta^*-\vartheta(t))d(a).
	\end{align*}
	By subtracting $\Theta^*$ from both sides of \eqref{eq1.33} simultaneously, it can be inferred that: 
	\begin{align*}
		&\vartheta(t+1)-\Theta^*\\
		&=\vartheta(t)-\Theta^*+\mathcal{H}(\vartheta(t),\Theta^k)^{-1}\\
		&\hspace{1cm}\int_{0}^{1}\mathcal{H}(\vartheta(t)+a(\Theta^*-\vartheta(t)),\Theta^k)(\Theta^*-\vartheta(t))d(a)\\
		&=\mathcal{H}(\vartheta(t),\Theta^k)^{-1}\int_{0}^{1}(\mathcal{H}(\vartheta(t)+a(\Theta^*-\vartheta(t)),\Theta^k)\\
		&\hspace{3.3cm}-\mathcal{H}(\vartheta(t),\Theta^k))(\Theta^*-\vartheta(t))d(a).
	\end{align*}
	Taking the norm on both sides of the relation above yields:
	\begin{align*}
		&\Vert\vartheta(t+1)-\Theta^*\Vert\\
		&\leq\Vert\mathcal{H}(\vartheta(t),\Theta^k)^{-1}\Vert\cdot\Vert\Theta^*-\vartheta(t)\Vert\\
		&\cdot\Vert\int_{0}^{1}(\mathcal{H}(\vartheta(t)+a(\Theta^*-\vartheta(t)),\Theta^k)-\mathcal{H}(\vartheta(t),\Theta^k))d(a)\Vert.
	\end{align*}
	According to the assumptions postulated on the Hessian matrix $\mathcal{H}$, we have: 
	\begin{align*}
		&\mathcal{H}(\vartheta(t),\Theta^k)\leq \frac{1}{\mathcal{A}},\\
		&\Vert\mathcal{H}(\vartheta_1,\Theta^k)-\mathcal{H}(\vartheta_2,\Theta^k)\Vert\leq \mathcal{B}\Vert\vartheta_1-\vartheta_2\Vert,
	\end{align*}
	for some nonnegative scalars $\mathcal A$ and $\mathcal B$. We can therefore derive the following bound on the distance between the current parameter estimate $\vartheta(t+1)$ and some $\Theta^*$:
	\begin{align*}
		\Vert\vartheta(t+1)-\Theta^*\Vert\leq\frac{2\mathcal{B}}{\mathcal{A}}\Vert\vartheta(t)-\Theta^*\Vert^2.
	\end{align*}
	For a finite training dataset of length $T$ we have hence proved that the sequence $\{\Theta^k\}_{k\in\mathbb{Z}}$ generated by Algorithm \ref{alg:EM} enjoys quadratic convergence to some $\Theta^*$.
\end{proof}

By assuming that the Hessian matrix is bounded and Lipschitz continuous, Theorem~\ref{th4} says that Algorithm~1 produces a sequence of parameter estimates that quickly converges to an optimal set of parameters $\Theta^*$. Generally speaking, if a function is strongly convex, then its Hessian matrix is usually bounded and Lipschitz continuous. These properties are 
%very important for the convergence analysis of optimization algorithms, as they can guarantee that the algorithm will not diverge during the iteration process and 
key for establishing converge to the optimal solution at a certain rate. 
%Therefore, verifying the strong convexity of the objective function is very useful when designing and analyzing optimization algorithms. This is typically done by calculating its second order derivatives \cite{Boyd2004}.

\section{Numerical experiments}\label{sec4}
The proposed EM-based technique is now tested on three numerical examples. Algorithm \ref{alg:EM} is then run $N$ times
% under the Assumptions \ref{assum:1}, \ref{assum:2} 
with $K$ fixed modes and the window length $T_w$. To evaluate the effectiveness of our methodology, we will make use of the mean square error (\textrm{MSE}):
\begin{align*}
	\textrm{MSE}=\frac{1}{T}\sum_{t=1}^{T}(y(t)-\hat{y}(t))^2,
\end{align*}
along with the best fit rate (\textrm{BFR}):
\begin{align*}
	\textrm{BFR}=100\left(1-\sqrt{\frac{\sum_{t=1}^{T}\Vert y(t)-\hat{y}(t)\Vert^2}{\sum_{t=1}^{T}\Vert y(t)-\overline{y}(t)\Vert^2}}\right)\%.
\end{align*}

\subsection{An academic example: a switching system composed of a linear and a nonlinear part}
Consider the 1000 input/output pairs generated by the system in \cite{Anna2018}, which consists of two submodels ($K=2$): one linear and another nonlinear. The value of inputs follows a uniform distribution and are hence randomly generated in $[0,1]$. The noise term follows a Gaussian distribution with zero mean and finite variance. The proposed NN structure with four RNN submodels is used to model the above switching system. Each of the RNN submodel has $l_f=l_g=2$ layers and four neurons. The activation function of these layers is $\arctan(\cdot)$. We then use the Algorithm \ref{alg:EM} to train the above switching system. 

In Fig.~\ref{fig:1} we show the true and estimated outputs, the resulting MSE, and the true switching sequence with the estimated one over the time window $[200,300]$. The three subgraphs presented in Fig.~\ref{fig:1} collectively illustrate the superior performance of the proposed method. Then, our technique is compared with the kernel-based approach proposed in \cite{Anna2018}. Table \ref{tab:3} reports both indices for the two methods under the different noise level, highlighting that our RNN-based technique exhibits better performance in most of the cases, especially in situations with large noise level.
%. From the Fig. \ref{fig:1} and Table \ref{tab:3}, we can find that the proposed method has better performance under the different noise condition. 
\begin{figure}[h]
	%	\hspace{-1.5cm}
	\includegraphics[height=0.8\linewidth]{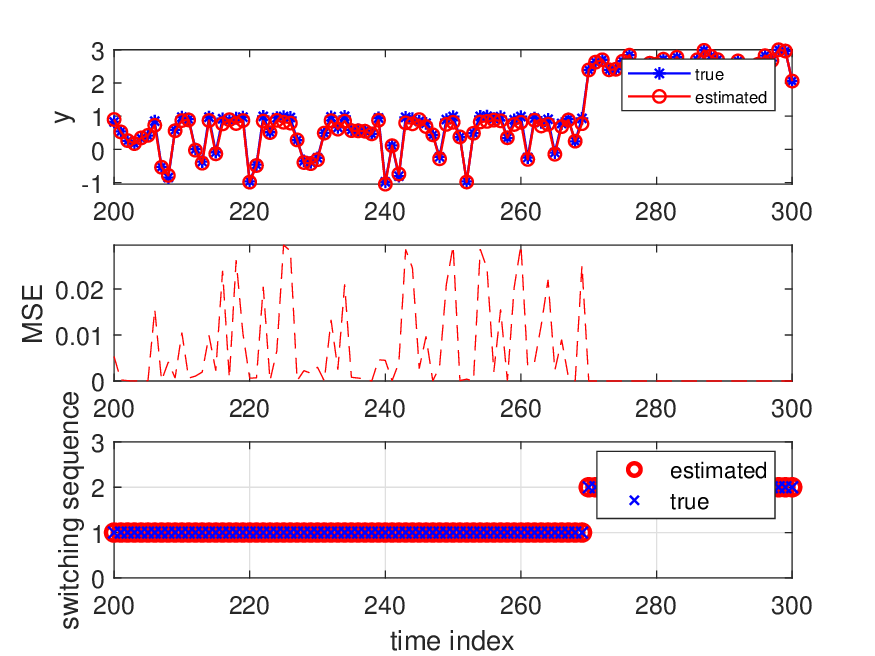} 
	\caption{Top: The true output (solid blue line with asterisks) and the estimated one by our method (solid red line with circles) over the time window $t\in[200,300]$. Middle: The resulting MSE. Bottom: The true switching sequence (blue crosses) and the estimated one (red circles).}
	\label{fig:1}
\end{figure}
\begin{table}[!t]
	\centering
	\setlength{\tabcolsep}{15pt}
	\renewcommand{\arraystretch}{1.4}
	\caption{The \textrm{MSE} and \textrm{BFR} for the proposed method and the kernel-based approach in \cite{Anna2018}, under different noise level.}
	\label{tab:3}
	\begin{tabular}{c|c|c|c|}
		&noise level&RNNs& \cite{Anna2018}\\\hline
		\multirow{4}{*}{\text{MSE}}& 0.001 & 0.0054&\textbf{0.0043} \\ \cline{2-4}
		& 0.01 &\textbf{ 0.0056} & 0.0684  \\ \cline{2-4}
		& 0.1 & \textbf{0.0155} & 0.0855 \\ \cline{2-4}
		& 0.2 &\textbf{0.0451} & 0.1061   \\ \hline
		\multirow{4}{*}{\text{BFR}}& 0.001 & $93.50\%$ & $\textbf{96.49\%}$  \\ \cline{2-4}
		& 0.01 & $\textbf{93.39\%}$ & $71.86\%$   \\ \cline{2-4}
		& 0.1 & $\textbf{89.01\%}$ & $69.22\%$   \\ \cline{2-4}
		& 0.2 & $\textbf{81.23}\%$&$67.22\%$  \\ \hline
	\end{tabular}
\end{table}
\subsection{An academic example: a switching system composed of two nonlinear parts}
Consider the following switching nonlinear system:
\begin{equation}\label{eq1.27}
	\begin{aligned}
		x(t+1)&=A_{s_t}\tanh(x(t))+B_{s_t}u(t)+\zeta(t),\\
		y(t)&=C_{s_t}\sin(x(t))+D_{s_t}u(t)-2+\xi(t),
	\end{aligned}
\end{equation}
with $n_u=1$, $n_x=3$, $n_y=1$. We take $x(0)=0$ and the values of the inputs $u(t)$ follows a uniform distribution and are hence randomly generated in $[0,1]$. The finite variance characterizing the noise terms $\zeta(t)$ and $\xi(t)$ is $Q_x=Q_y=10^{-3}$. We consider $K=2$ modes with the following system matrices/vectors:
\begin{align*}
	A_1&=\begin{bmatrix}
		0.8 & 0.2& -0.1\\
		0& 0.9& 0.1\\
		0.1& -0.1& 0.7
	\end{bmatrix}, &A_2&=\begin{bmatrix}
		0.5& -0.2& -0.1\\
		0& 0.9& 0.1\\
		-0.1& -0.3& 0.8
	\end{bmatrix},\\
	B_1&=[-1 \; 0.5 \; 1]^\top, &B_2&=[-0.5 \; 0.1 \; 0.5]^\top,\\
	C_1&=[-1 \; 1.5 \; 0.5], &C_2&=[-0.1 \; -0.5 \; 0.8],\\
	D_1&=0.1,&D_2&=-0.1,\\
	\Pi&=\begin{bmatrix}
		0.98& 0.02\\
		0.02& 0.98
	\end{bmatrix}.
\end{align*}
The data are generated by \eqref{eq1.27} for $T=1000$ samples. Then, we use the RNNs in \eqref{eq1.5} to model the dynamics in \eqref{eq1.27}, where each RNN submodel $\mathcal{N}_{x,1}$, $\mathcal{N}_{x,2}$, $\mathcal{N}_{y,1}$, and $\mathcal{N}_{y,2}$ has $l_f=l_g=2$ layers with six neurons. The activation function for the first layer $f_1$, $g_1$ is $\arctan(\cdot)$, and the output layer $g_2$ is instead a linear mapping. Thus, the parameters $\theta_{f,1}$, $\theta_{f,2}$, $\theta_{g,1}$, $\theta_{g,2}$, and the switching sequence $\bm{S}$ fully describe the system dynamics in \eqref{eq1.5}.

We run Algorithm~\ref{alg:EM} for $N=10$ times and the EKF process for $10$ epochs, for each RNN submodel.
%To compare the quality of the estimation results, 
In this example, our technique is compared with the Bayesian ensemble learning algorithm proposed in \cite{mlp}. Fig.~\ref{fig:3} shows the true and estimated outputs obtained by the proposed method, demonstrating that our EM-based algorithm can achieve a great performance in the parameter estimation, with a remarkable $100\%$ match rate in estimating the switching sequence.

To verify the performance of the proposed method against different noise levels, we use the trained RNNs to predict $100$ trajectories with the different initial states and noise conditions. In Figs. \ref{fig:6} and \ref{fig:7} we report both mean and variance of the two indices considered, comparing the results obtained with our RNN-based method, the kernel-based method \cite{Anna2018} and Bayesian ensemble learning \cite{mlp}. It can indeed be observed that our technique features better accuracy and lower error, since our method has excellent performance in both MSE and BFR evaluation metrics, regardless of the level of noise. 

\begin{figure}[t]
	\hspace{-0.5cm}
	\includegraphics[height=0.85\linewidth]{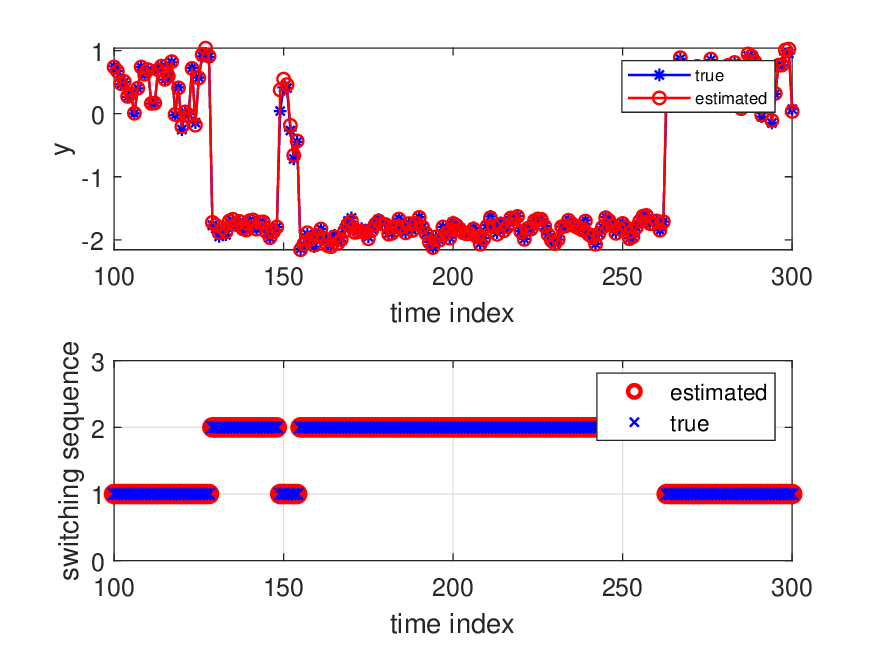} 
	\caption{Top: The true (solid blue line with asterisks) and the estimated output by using our RNN-based method (solid red line with circles) over the time window $t\in[100,300]$. Bottom: The true switching sequence (blue crosses) and the one estimated by our method (red circles).}
	\label{fig:3}
\end{figure}
\begin{figure}[h]
	\centering
	\includegraphics[height=0.8\linewidth]{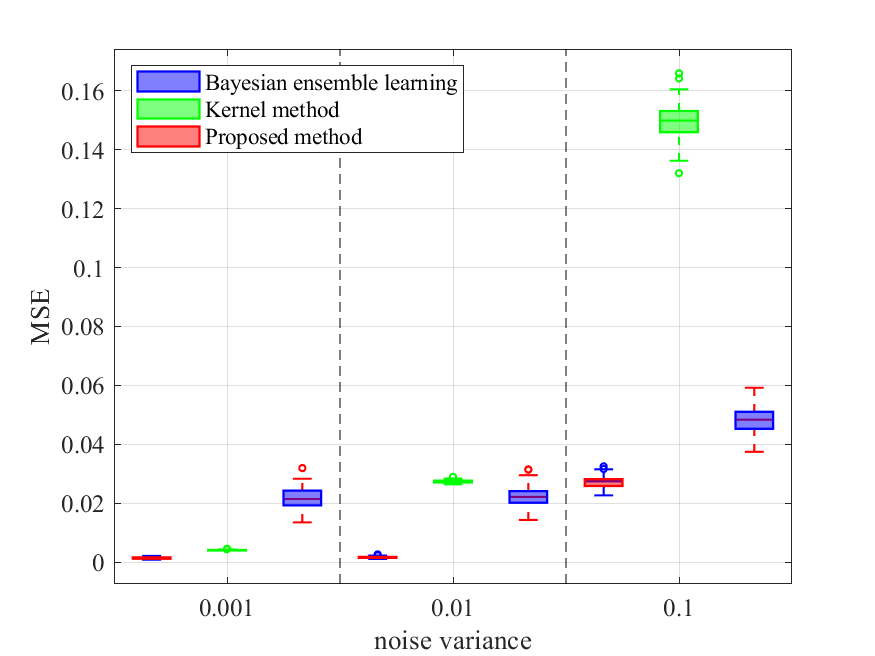} 
	\caption{The MSE obtained by proposed method (red), the kernel-based method (yellow) \cite{Anna2018} and the Bayesian ensemble learning (blue) \cite{mlp} under different noise conditions.}
	\label{fig:6}
\end{figure}
\begin{figure}[h]
	\centering
	\includegraphics[height=0.8\linewidth]{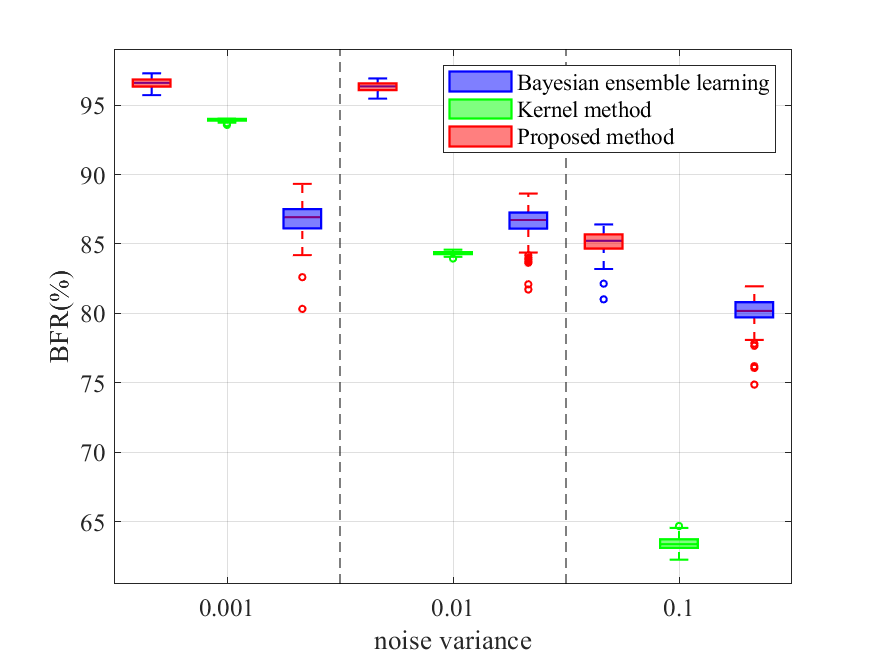} 
	\caption{The BFR obtained by proposed method (red), the kernel-based method (yellow) \cite{Anna2018} and the Bayesian ensemble learning (blue) \cite{mlp} under different noise conditions.}
	\label{fig:7}
\end{figure}
%\begin{table*}[!t]
%	\centering
%	\setlength{\tabcolsep}{10pt}
%	\caption{The mean and variance of \textrm{MSE} and \textrm{BFR} for the two algorithms under different noise conditions.}
%	\label{tab:2}
%	\begin{tabular}{|c|c|c|c|c|c|c|}\hline
	%		method&noise level&BFR (mean)&BFR (variance)&MSE (mean)&MSE (variance)\\\hline
	%		\multirow{4}{*}{\text{RNN}}& 0.001 & $96.58\%$&$1.233\times10^{-5}$  & 0.0015 &$8.245\times10^{-8}$ \\ \cline{2-6}
	%		& 0.01 & $96.31\%$ & $1.107\times10^{-5}$ & 0.0016 &$8.143\times10^{-8}$ \\ \cline{2-6}
	%		& 0.01 & $85.06\%$ & $7.424\times10^{-5} $& 0.0272 &$3.097\times10^{-6}$ \\ \cline{2-6}
	%		& 0.1 & $70.50\%$ & $1.785\times10^{-4}$ & 0.1064 &$3.559\times10^{-5}$ \\ \hline
	%		\multirow{4}{*}{\cite{mlp}}& 0.001 & $86.73\%$ & $1.572\times10^{-4}$ & 0.0217 & $1.322\times10^{-5}$\\ \cline{2-6}
	%		& 0.01 & $86.50\%$ & $1.599\times10^{-4}$ & 0.0223 &$1.337\times10^{-5}$ \\ \cline{2-6}
	%		& 0.1 & $80.04\%$ & $1.389\times10^{-4}$ & 0.0483&$1.736\times10^{-5}$ \\ \cline{2-6}
	%		& 0.2 & $68.03\%$&$2.615\times10^{-4}$ & 0.1240 &$6.207\times10^{-5}$ \\ \hline
	%	\end{tabular}
%\end{table*}

\subsection{Battery state of charge estimation}
We now test our approach to estimate the state of charge (SOC) in a battery management system. Specifically, we use a battery platform to generate the dataset of $T=1047$ samples, which consists of a computer, a battery tester (NEWARE CT-4008-5V12A-TB) and a group of lithium batteries (NCR18650PF). The dataset is composed of two parts: the control input consisting of both current and voltage in the circuit, and the output including the SOC at each time instant. In addition, the data contains two operating modes of the circuit overall, i.e., charging and discharging, illustrated in Fig.~\ref{fig:4}. Therefore, we use $K=2$ submodels to characterize different circuit operations. Then, RNNs are used to model the relation between inputs (current and voltage) and outputs (SOC) in each of these operating modes. Specifically, we choose $l_f=l_g=2$ layers with six neurons each. The activation functions adopted in the first and the second layer are standard affine and sigmoid functions, respectively. The EKF-based training procedure is run for $20$ epochs for each RNN submodel, and Algorithm~\ref{alg:EM} is run $N=10$ times. The SOC estimates and MSE errors are shown in Fig. \ref{fig:5} which reports the great performance of the proposed method. Furthermore, the match rate of the switching sequence achieves $100\%$.
\begin{figure}
	\centering
	\includegraphics[height=0.8\linewidth]{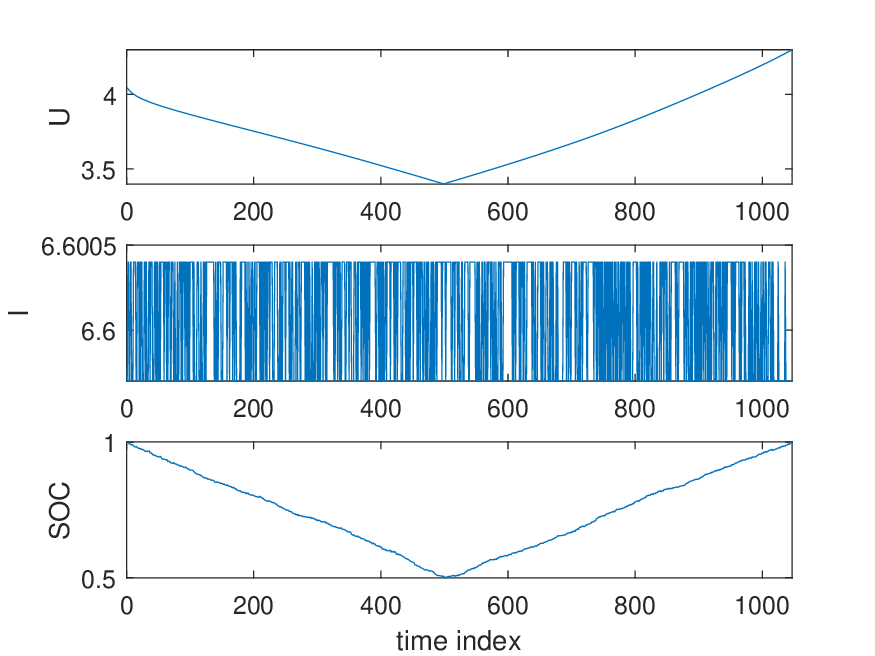} 
	\caption{The training data collected in the battery platform where $U$ and $I$ respectively represent the voltage and current in the circuit, and SOC represents the battery state of charge.}
	\label{fig:4}
\end{figure}
\begin{figure}
	\centering
	\includegraphics[height=0.8\linewidth]{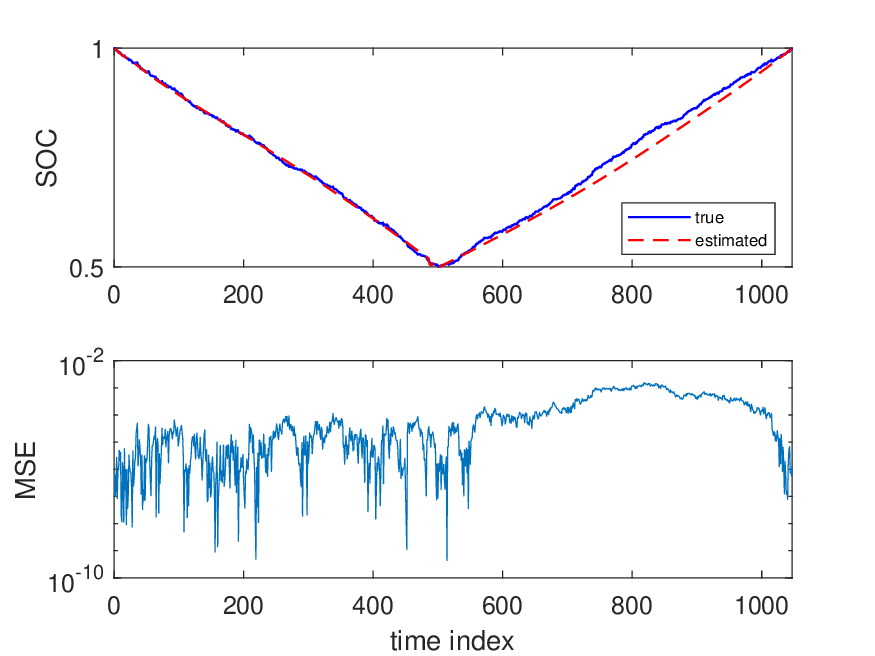} 
	\caption{Estimates and MSE errors of the SOC by using the proposed method.}
	\label{fig:5}
\end{figure}

We finally compare our technique with a standard SOC learning method \cite{Jiao}, which uses the gate recurrent unit (GRU) based momentum algorithm to estimate the SOC. 
Table~\ref{tab:1} reports both \textrm{MSE} and \textrm{BFR} indices for the two methods, highlighting how our approach results in better performance.

\begin{table}[!hbtp]
	\centering
	\setlength{\tabcolsep}{15pt}
	\renewcommand{\arraystretch}{1.4}
	\caption{Battery platform: The \textrm{MSE} and \textrm{BFR} for the two algorithms}
	\label{tab:1}
	\begin{tabular}{lll}\hline
		Method & \textrm{BFR}&  \textrm{MSE}  \\\hline
		GRU \cite{Jiao}& 88.25\%&0.0020\\\hline 
		Proposed method& \textbf{92.62\%} &\textbf{0.0004}\\\hline
	\end{tabular}
\end{table}

\section{Conclusion}\label{sec5}

We have presented a neural network-based scheme for the identification of switching nonlinear systems with unknown inner structures. In the EM framework, we have then devised an iterative procedure alternating an E-step and M-step. While in the former we have used a moving window approach to obtain the posterior estimate of the switching sequence, with significant reduction of the overall computational complexity, in the latter an EKF-based training procedure, suitable for switching systems, has been proposed to retrieve the parameter estimates for each subsystem, modeled through an RNN. Numerical experiments have been carried out to demonstrate that the proposed methodology exhibits excellent performance in terms of parameter estimation, model fitting, and switching sequence identification.

Besides RNNs, more expressive neural network architectures, such as long short-term memory networks (LSTMs) or GRUs, could be considered to further enhance the model's generalization ability. In addition, one may also investigate the case in which the number of subsystems is not known, thereby facing further technical challenges. These aspects are left to future work. 

%
%
%
%%\begin{ack}                               % Place acknowledgements
%%  % here.
%%\end{ack}
%
%%\bibliographystyle{plain}        % Include this if you use bibtex 
%%\bibliography{autosam}           % and a bib file to produce the 
%%                                 % bibliography (preferred). The
%%                                 % correct style is generated by
%%                                 % Elsevier at the time of printing.
%

\bibliographystyle{plain}
\bibliography{biblio}

\appendix
\section{Technical proofs}\label{sec:proofs}
\textit{Proof of Proposition~\ref{prop1}:} In view of Standing Assumption~\ref{assum:2}, we have:
\begin{align*}
	&\mathbb{P}[\bm{y},\bm{x},\bm{S},\Theta \vert\bm{u}]\\
	&=\mathbb{P}[\bm{y},\bm{x}\vert\bm{u},\bm{S},\Theta]\mathbb{P}[\bm{S},\Theta\vert\bm{u}]\\
	&=\mathbb{P}[\bm{y}\vert\bm{x},\bm{u},\bm{S},\bm{\theta}_g^k]\mathbb{P}[\bm{x}\vert,\bm{u},\bm{S},\bm{\theta}_f^k]\\
	&\hspace{4.5cm}\mathbb{P}[\bm{S}\vert\bm{u}]\mathbb{P}[\bm{\theta}_{f}^k\vert\bm{u}]\mathbb{P}[\bm{\theta}_{g}^k\vert\bm{u}]\\
	&=\mathbb{P}[\bm{y}\vert\bm{x},\bm{u},\bm{S},\bm{\theta}_g^k]\mathbb{P}[\bm{x}\vert\bm{u},\bm{S},\bm{\theta}_f^k]\mathbb{P}[\bm{S}]\mathbb{P}[\bm{\theta}_{f}^k]\mathbb{P}[\bm{\theta}_{g}^k]\\
	&=\prod_{t=1}^{T}\mathbb{P}[y(t)\vert x(t),u(t),\theta_{g,s_t}]\\
	&\hspace{1.5cm}\prod_{t=1}^{T}\mathbb{P}[x(t)\vert x(t-1), u(t-1),\theta_{f,s_{t}}]\prod_{i=1}^{K}\mathbb{P}[\theta_{f,i}]\\
	&\hspace{1.5cm}\prod_{i=1}^{K}\mathbb{P}[\theta_{g,i}]\mathbb{P}[s_0]\prod_{t=1}^{T}\mathbb{P}[s_t\vert s_{t-1}].
\end{align*}
By taking the logarithm of the probability density function as mentioned in \S \ref{sec3} we then obtain:
\begin{align*}
	&\log\mathbb{P}[\bm{y},\bm{x},\bm{S},\Theta \vert\bm{u}]\\
	&=\sum_{t=1}^{T}\log\mathbb{P}[y(t)\vert x(t),u(t),\theta_{g,s_t}]\\
	&+\sum_{t=1}^{T}\log\mathbb{P}[x(t)\vert x(t-1), u(t-1),\theta_{f,s_{t}}]+\sum_{i=1}^{K}\log\mathbb{P}[\theta_{f,i}]\\
	&+\sum_{i=1}^{K}\log\mathbb{P}[\theta_{g,i}]+\log\mathbb{P}[s_0]+\sum_{t=1}^{T}\log\mathbb{P}[s_t\vert s_{t-1}].
\end{align*}
Then, minimizing $J(\bm{y},\bm{u},\Theta,\bm{S})$ in \eqref{eq1.7} w.r.t. $\Theta$, whose components are defined in \eqref{subeq1.6}, it is equivalent to maximizing the logarithm of $\mathbb{P}[\bm{y},\bm{x},\bm{S},\Theta \vert\bm{u}]$, which is hence also equivalent to maximizing $\mathbb{P}[\bm{y},\bm{x},\bm{S},\Theta]$.
\qedsymbol

\emph{Proof of Theorem~\ref{th1}:}
From Proposition \ref{prop1}, we know that minimizing the loss function $J(\bm{y},\bm{u},\Theta,\bm{S})$ defined in \eqref{eq1.7} with ingredients as in \eqref{subeq1.6} is equivalent to maximizing $Q(\Theta,\Theta^k)$ w.r.t. $\Theta$. Therefore, we only need to prove that applying the EKF process to each submodel minimizes the loss function $J(\bm{y},\bm{u},\Theta,\bm{S})$. For conciseness, only one iteration process of EKF is considered, while the others can be derived in the same way.

At the $n$-th step, the loss function can be written as 
\begin{align}\label{eq1.20}
	&J(\bm{y}_{1:n},\bm{u}_{1:n},\Theta,\bm{S}_{1:n})\nonumber\\
	&=\sum_{t=1}^{n}\ell(y(t),u(t),\theta_{f,s_t},\theta_{g,s_t})+\sum_{i=1}^{K}(r(\theta_{f,i})+r(\theta_{g,i}))\nonumber\\
	&\hspace{7cm}+\mathcal{L}(\bm{S})\nonumber\\
	&=J(\bm{y}_{1:n-1},\bm{u}_{1:n-1},\Theta^k,\bm{S}_{1:n-1})\nonumber\\
	&\hspace{1cm}+\ell(y(n),u(n),\theta_{f,s_n},\theta_{g,s_n})+\mathcal{L}(s_n).
\end{align}
Let us introduce $z_n\coloneqq\{\bm{y}_{1:n},\bm{u}_{1:n},\Theta$, $\bm{S}_{1:n}\}$. The loss function in \eqref{eq1.20} can hence be rewritten as
\begin{align*}
	J(z_n)=J(z_{n-1})+\ell(y(n),u(n),\theta_{f,s_n},\theta_{g,s_n})+\mathcal{L}(s_n).
\end{align*} 

The EKF is a recursive process consisting in two main parts. The first part consists in predicting the prior estimate of the parameter in \eqref{eq1.12}. Similar to \S \ref{sec3.1}, we indicate with $\hat{z}_n^-$ the prior estimate of $z_n$ so that the prior loss can be rewritten as: 
\begin{align}\label{eq1.21}
	J^-(z_n)&=J(z_{n-1})+\ell(x(n+1),x(n),u(n),\theta_{f,s_n})\nonumber\\
	&\hspace{4cm}+\frac{1}{2}\mathcal{L}(s_n)
\end{align}
In the second part, one updates the prior estimate by using the measurement data \eqref{eq1.13}. To this end, let $\hat{z}_n$ be the posterior estimate of $z_n$. Then, the posterior loss reads as:
\begin{align}\label{eq1.22}
	J(z_n)=J^-(z_{n})+\ell(y(n),x(n),u(n),\theta_{g,s_n})+\frac{1}{2}\mathcal{L}(s_n).
\end{align}
Let us then consider the prediction step first.
To apply the EKF, one has to linearize the nonlinear system \eqref{eq1.9} as in \eqref{eq1.10} and \eqref{eq1.11}. Thus, the RNN-based system is locally equivalent to the following linear system:
\begin{align*}
	x(t+1)&=F(t)\nu(t)+\tilde{u}_x(t)+\zeta(t),\\
	y(t)&=H(t)\nu(t)+\tilde{u}_y(t)+\xi(t),\\
	\nu(t)&=\begin{bmatrix}
		x(t)\\
		\vartheta(t)
	\end{bmatrix},
\end{align*}
where $\tilde{u}_x(t)=\mathcal{N}_{x,s_t}(x(t),u(t),\theta_{f,s_t})-F(t)\nu(t)$, and $\tilde{u}_y(t)=\mathcal{N}_{y,s_t}(x(t),u(t),\theta_{g,s_t})-H(t)\nu(t)$. 
A standard method for solving linear equations is adopting the weighted least-square, which relies on the loss function $	\ell(x(n+1),x(n),u(n),\theta_{f,s_{n}})=\frac{1}{2}\Vert x(n+1)-\mathcal{N}_{x,s_n}(x(n),u(n),\theta_{f,s_n})\Vert_{Q_x^{-1}}$, where $Q_x$ is a positive definite weight matrix (notice the abuse of notation with the $\ell(\cdot)$ determining \eqref{eq1.7}). To alleviate notation, we omit the dependency for $\mathcal{N}_{x,s_n}$ on its arguments, and define $\bm{\nu}_{1:n}\coloneqq[\nu(1) \ldots \nu(n)]^\top$. Then, deriving the expression in \eqref{eq1.21} w.r.t. $\bm{\nu}_{1:n+1}$ leaves us with
%	 \begin{small}
	\begin{align*}
		\frac{\partial J^-(z_n)}{\partial\bm{\nu}_{1:n+1}}=\begin{bmatrix}
			\frac{\partial J(z_{n-1})}{\partial\bm{\nu}_{1:n}}
			-F(n)^\top Q_x^{-1}(x(n+1)-\mathcal{N}_{x,s_n})\\
			Q_x^{-1}(x(n+1)-\mathcal{N}_{x,s_n})\end{bmatrix}.
	\end{align*}
	%	 \end{small}
Let the gradient be zero (i.e., $\nu(n)=\hat{\nu}^-(n)$), which yields \eqref{eq1.12}. We derive next the covariance of the estimate $\hat{\nu}(n+1)$. Specifically, the Hessian matrix reads as:
\begin{align*}
	&\frac{\partial^2 J^-(z_n)}{\partial\bm{\nu}_{1:n+1}^2}\\
	&=\begin{bmatrix}
		\frac{\partial^2 J(z_{n-1})}{\partial\bm{\nu}_{1:n}^2}+F(n)Q_x^{-1}F^\top(n) &-F(n)^\top Q_x^{-1}\\
		-Q_x^{-1}F(n)&Q_x^{-1} \end{bmatrix},
\end{align*}
which in view of its positive definiteness, the one-step iteration of the Newton's method can be obtained through the gradient and the Hessian matrix:
\begin{align*}
	\nu(n+1)=\nu(n)-\left(\frac{\partial^2 J^-(z_n)}{\partial\bm{\nu}_{1:n+1}^2}\right)^{-1}\frac{\partial J^-(z_n)}{\partial\bm{\nu}_{1:n+1}},
\end{align*}
where the lower-right block of the inverse of the Hessian matrix is the covariance matrix of the prior estimation error. Then, we have $P^-(n)=F(n)P(n-1)F(n)^\top+Q_x$, which yields \eqref{eq1.15} with $Q_x(t)=\mathrm{diag}(\Sigma_1(t),\Sigma_\vartheta(t))$.

Successively, we need to prove that the update part of the EKF process also minimizes the loss function in \eqref{eq1.22}. Similar to the prediction part, the gradient and the Hessian matrix of $J(\hat{z}_n)$ reads as:
\begin{align*}
	\frac{\partial J(z_n)}{\partial\bm{\nu}_{1:n}}&=
	\frac{\partial J^-(z_n)}{\partial\bm{\nu}_{1:n}}-H(n)^\top Q_y^{-1}(y(n)-\mathcal{N}_{y,s_n}), \\
	\frac{\partial^2 J(z_n)}{\partial\bm{\nu}_{1:n}^2}&=
	\frac{\partial^2 J^-(z_{n})}{\partial\bm{\nu}_{1:n}^2}+H(n)Q_y^{-1}H^\top(n).
\end{align*}
Let $\nu(n)=\hat{\nu}^-(n)$. The gradient then takes the form:
\begin{align}\label{eq1.31}
	\frac{\partial J(z_n)}{\partial\bm{\nu}_{1:n}}=\begin{bmatrix}
		0\\
		-H(n)^\top Q_y^{-1}(y(n)-\mathcal{N}_{y,s_n})
	\end{bmatrix}.
\end{align}
Similarly, the covariance matrix of the posterior estimation error can be expressed as the right-block of the inverse of the Hessian matrix, which reads as follows:
\begin{align}\label{eq1.32}
	P(n)=(P^-(n)+H(n)^\top Q_yH(n))^{-1}.
\end{align}
\sloppy Let us denote the Kalman gain as $\Gamma(n)=P^-(n)H^\top(n)[H(n)P^-(n)H^\top(n)+Q_y]^{-1}$.
Then, the covariance matrix is $P(n)$ equivalent to $(I-\Gamma(n)H(n))P^-(n)$ which yields \eqref{eq1.18}.

Therefore, the EKF process minimizes the one-step loss function \eqref{eq1.21} and \eqref{eq1.22}, and in turn minimizes the loss function \eqref{eq1.20} jointly. This concludes the proof.
\qedsymbol

%\emph{Proof of Theorem~\ref{th3}:} 
%\qedsymbol
\bibliographystyle{ieeetr}
\bibliography{references}

\begin{thebibliography}{10}

\bibitem{Amos2017}
Brandon Amos and J~Zico Kolter.
\newblock Optnet: Differentiable optimization as a layer in neural networks.
\newblock In {\em International conference on machine learning}, pages
  136--145. PMLR, 2017.

\bibitem{Bako2011}
Laurent Bako.
\newblock Identification of switched linear systems via sparse optimization.
\newblock {\em Automatica}, 47(4):668--677, 2011.

\bibitem{Mark2022}
Mark~P. Balenzuela, Adrian~G. Wills, Renton Christopher, and Ninness Brett.
\newblock Parameter estimation for jump markov linear systems.
\newblock {\em Automatica}, 135:109949, 2022.

\bibitem{Batruni1991}
Roy Batruni.
\newblock A multilayer neural network with piecewise-linear structure and
  back-propagation learning.
\newblock {\em IEEE Transactions on Neural Networks}, 2(3):395--403, 1991.

\bibitem{Bemporad2023}
Alberto Bemporad.
\newblock Recurrent neural network training with convex loss and regularization
  functions by extended {Kalman} filtering.
\newblock {\em IEEE Transactions on Automatic Control}, 68(9):5661--5668, 2022.

\bibitem{Bemporad2018}
Alberto Bemporad, Valentina Breschi, Dario Piga, and Stephen~P Boyd.
\newblock Fitting jump models.
\newblock {\em Automatica}, 96:11--21, 2018.

\bibitem{Piga2016b}
Valentina Breschi, Dario Piga, and Alberto Bemporad.
\newblock Piecewise affine regression via recursive multiple least squares and
  multicategory discrimination.
\newblock {\em Automatica}, 73:155--162, 2016.

\bibitem{Carloni2007}
Raffaella Carloni, Ricardo~G Sanfelice, Andrew~R Teel, and Claudio Melchiorri.
\newblock A hybrid control strategy for robust contact detection and force
  regulation.
\newblock In {\em 2007 American Control Conference}, pages 1461--1466. IEEE,
  2007.

\bibitem{Chan2008}
Antoni~B Chan and Nuno Vasconcelos.
\newblock Modeling, clustering, and segmenting video with mixtures of dynamic
  textures.
\newblock {\em IEEE transactions on pattern analysis and machine intelligence},
  30(5):909--926, 2008.

\bibitem{Fabiani2025}
Filippo Fabiani, Bartolomeo Stellato, Daniele Masti, and Paul~J Goulart.
\newblock A neural network-based approach to hybrid systems identification for
  control.
\newblock {\em Automatica}, 2025.
\newblock (In press).

\bibitem{Ferrari2003}
Giancarlo Ferrari-Trecate, Marco Muselli, Diego Liberati, and Manfred Morari.
\newblock A clustering technique for the identification of piecewise affine
  systems.
\newblock {\em Automatica}, 39(2):205--217, 2003.

\bibitem{survey}
Andrea Garulli, Simone Paoletti, and Antonio Vicino.
\newblock A survey on switched and piecewise affine system identification.
\newblock {\em IFAC Proceedings Volumes}, 45(16):344--355, 2012.

\bibitem{Golabi2017}
Arash1 Golabi, Nader1 Meskin, Roland Toth, and Javad Mohammadpour.
\newblock A bayesian approach for lpv model identification and its application
  to complex processes.
\newblock {\em IEEE Transactions on Control Systems Technology},
  25(6):2160--2167, 2017.

\bibitem{Guidolin}
Massimo Guidolin.
\newblock Markov switching models in empirical finance.
\newblock In {\em Missing data methods: Time-series methods and applications},
  pages 1--86. Emerald Group Publishing Limited, 2011.

\bibitem{mlp}
Antti Honkela.
\newblock {\em Nonlinear switching state-space models}.
\newblock PhD thesis, Helsinki University of Technology, 2001.

\bibitem{Hornik1991}
Kurt Hornik.
\newblock Approximation capabilities of multilayer feedforward networks.
\newblock {\em Neural networks}, 4(2):251--257, 1991.

\bibitem{Hornik1989}
Kurt Hornik, Maxwell Stinchcombe, and Halbert White.
\newblock Multilayer feedforward networks are universal approximators.
\newblock {\em Neural networks}, 2(5):359--366, 1989.

\bibitem{Humpherys2012}
Jeffrey Humpherys, Preston Redd, and Jeremy West.
\newblock A fresh look at the {Kalman} filter.
\newblock {\em SIAM review}, 54(4):801--823, 2012.

\bibitem{Jiao}
Meng Jiao, Dongqing Wang, and Jianlong Qiu.
\newblock A {GRU}-{RNN} based momentum optimized algorithm for {SOC}
  estimation.
\newblock {\em Journal of Power Sources}, 459:228051, 2020.

\bibitem{Toth2024}
Kon Johan, Roland Toth, Wijdeven Jeroen, Heertjes Marcel, and Oomen Tom.
\newblock Guaranteeing stability in structured input-output models: With
  application to system identification.
\newblock {\em IEEE Control Systems Letters}, 8:1565--1570, 2024.

\bibitem{Lauer2015}
Fabien Lauer.
\newblock On the complexity of piecewise affine system identification.
\newblock {\em Automatica}, 62:148--153, 2015.

\bibitem{Lauer2011}
Fabien Lauer, G{\'e}rard Bloch, and Ren{\'e} Vidal.
\newblock A continuous optimization framework for hybrid system identification.
\newblock {\em Automatica}, 47(3):608--613, 2011.

\bibitem{Bemporad2021}
Daniele Masti and Alberto Bemporad.
\newblock Learning nonlinear state--space models using autoencoders.
\newblock {\em Automatica}, 129:109666, 2021.

\bibitem{Ohlsson2010}
Henrik Ohlsson, Lennart Ljung, and Stephen Boyd.
\newblock Segmentation of {ARX}-models using sum-of-norms regularization.
\newblock {\em Automatica}, 46(6):1107--1111, 2010.

\bibitem{Piga2020}
Dario Piga, Valentina Breschi, and Alberto Bemporad.
\newblock Estimation of jump {Box}--{Jenkins} models.
\newblock {\em Automatica}, 120:109126, 2020.

\bibitem{Piga2015}
Dario Piga, Pepijn Cox, Roland Toth, and Vincent Laurain.
\newblock Lpv system identification under noise corrupted scheduling and output
  signal observations.
\newblock {\em Automatica}, 53(C):329--338, 2015.

\bibitem{Porreca2009}
Riccardo Porreca, Samuel Drulhe, Hidde De~Jong, and Giancarlo Ferrari-Trecate.
\newblock Identification of parameters and structure of piecewise affine models
  of genetic networks.
\newblock {\em IFAC Proceedings Volumes}, 42(10):587--592, 2009.

\bibitem{Prasad2003}
Vinay Prasad and B~Wayne Bequette.
\newblock Nonlinear system identification and model reduction using artificial
  neural networks.
\newblock {\em Computers \& Chemical Engineering}, 27(12):1741--1754, 2003.

\bibitem{Ripaccioli2009}
Giulio Ripaccioli, Alberto Bemporad, Francis Assadian, Clement Dextreit,
  Stefano Di~Cairano, and Ilya~V Kolmanovsky.
\newblock Hybrid modeling, identification, and predictive control: An
  application to hybrid electric vehicle energy management.
\newblock In {\em Hybrid Systems: Computation and Control: 12th International
  Conference, HSCC 2009, San Francisco, CA, USA, April 13-15, 2009. Proceedings
  12}, pages 321--335. Springer, 2009.

\bibitem{Roll2004}
Jacob Roll, Alberto Bemporad, and Lennart Ljung.
\newblock Identification of piecewise affine systems via mixed-integer
  programming.
\newblock {\em Automatica}, 40(1):37--50, 2004.

\bibitem{Anna2018}
Anna Scampicchio, Alberto Giaretta, and Gianluigi Pillonetto.
\newblock Nonlinear hybrid systems identification using kernel-based
  techniques.
\newblock volume~51, pages 269--274. Elsevier, 2018.

\bibitem{Schlegl2003}
Thomas Schlegl, Martin Buss, and G{\"u}nther Schmidt.
\newblock A hybrid systems approach toward modeling and dynamical simulation of
  dextrous manipulation.
\newblock {\em IEEE/ASME transactions on mechatronics}, 8(3):352--361, 2003.

\bibitem{Schuller2008}
Bj{\"o}rn Schuller, Martin W{\"o}llmer, Tobias Moosmayr, G{\"u}nther Ruske, and
  Gerhard Rigoll.
\newblock Switching linear dynamic models for noise robust in-car speech
  recognition.
\newblock In {\em Pattern Recognition: 30th DAGM Symposium Munich, Germany,
  June 10-13, 2008 Proceedings 30}, pages 244--253. Springer, 2008.

\bibitem{Ji2022}
Ji~Shuyi, Zhang Zizhao, Ying Shihui, Wang Liejun, Zhao Xibin, and Yue Gao.
\newblock Kullback–leibler divergence metric learning.
\newblock {\em IEEE transactions on cybernetics}, 52(4):2047--2058, 2022.

\bibitem{Tan2023}
Kelvin Tan, William~J Parquette, and Meng Tao.
\newblock A predictive algorithm for maximum power point tracking in solar
  photovoltaic systems through load management.
\newblock {\em Solar Energy}, 265:112127, 2023.

\bibitem{Timmermann}
Allan Timmermann.
\newblock {\em Markov Switching Models in Finance}, pages 1--3.
\newblock John Wiley \& Sons, Ltd, 2015.

\bibitem{Umenberger2018}
J~Umenberger, J~Wågberg, I.R Manchester, and T.B. Schon.
\newblock Maximum likelihood identification of stable linear dynamical systems.
\newblock {\em Automatica}, 96:280--292, 2018.

\bibitem{Vidal2002}
Ren{\'e} Vidal, Alessandro Chiuso, and Stefano Soatto.
\newblock Observability and identifiability of jump linear systems.
\newblock In {\em Proceedings of the 41st IEEE Conference on Decision and
  Control, 2002.}, volume~4, pages 3614--3619. IEEE, 2002.

\bibitem{Xu2009}
Jun Xu, Xiaolin Huang, and Shuning Wang.
\newblock Adaptive hinging hyperplanes and its applications in dynamic system
  identification.
\newblock {\em Automatica}, 45(10):2325--2332, 2009.

\bibitem{Pan2024}
Pan Zhaojie, Li~Chenyu, Plaza Antonio, Chanussot Jocelyn, and Hong Danfeng.
\newblock Hyperspectral image classification with mamba.
\newblock {\em IEEE Transactions on Geoscience and Remote Sensing}, page~1,
  2024.

\end{thebibliography}

%\begin{IEEEbiography}[{\includegraphics[width=0.95in,height=1.3in,clip]{Junyao.jpg}}]%
%	{Junyao You}received the B.S. and M.S. degrees from
%	the School of Internet of Things Engineering, Jiangnan
%	University, Wuxi, China, in 2017 and 2020, respectively.
%	She is currently pursuing the Ph.D. degree with the
%	School of Automation, Beijing Institute of Technology,
%	Beijing, China. Her current research interests include
%	topology identification and parameter estimation of
%	dynamical networks.
%\end{IEEEbiography}
\begin{IEEEbiography}[{\includegraphics[width=0.95in,height=1.3in,clip]{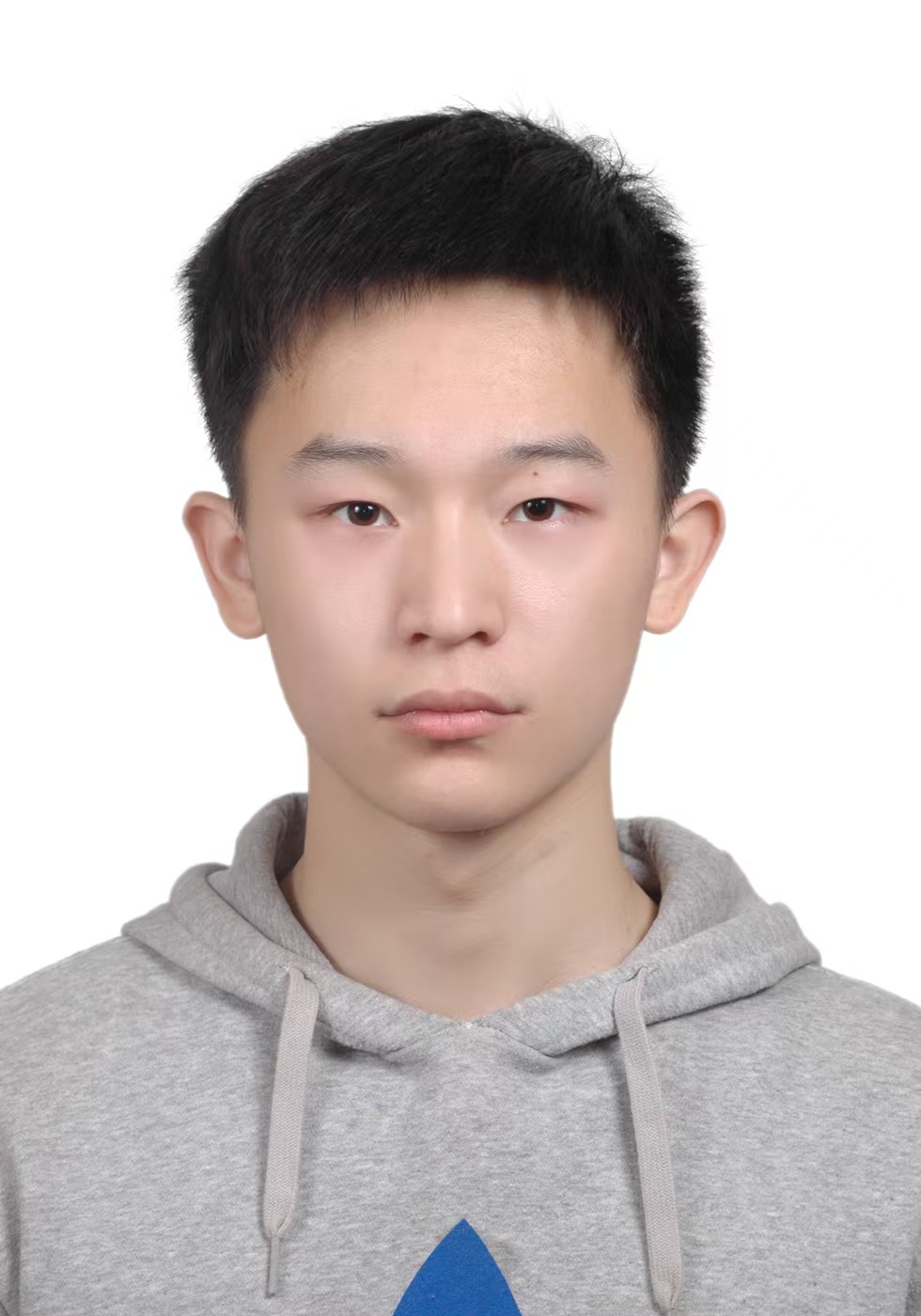}}]%
	{Yanxin Zhang} received his B.Sc. degree in the School of Mathematics and Statistics from the Huazhong University of Science and Technology, Wuhan, China, in 2019, and the M.Sc. degree in
	the School of Science from Jiangnan University, Wuxi, China, in 2023. He is currently pursuing the Ph.D. degree with the
	School of Automation, Beijing Institute of Technology,
	Beijing, China. His current research interests include
	system identification and parameter estimation of
	switching systems.
\end{IEEEbiography}
\begin{IEEEbiography}[{\includegraphics[width=0.95in,height=1.2in,clip]{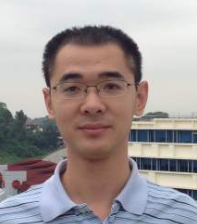}}]
{Chengpu Yu} received the B.E. and M.E. degrees in electrical engineering from the University of Electronic Science and Technology
of China, Chengdu, China, in 2006 and 2009,
respectively, and the Ph.D. degree in electrical
engineering from Nanyang Technological University, Singapore, in 2014.
He was with the Internet of Things Lab at
Nanyang Technological University as a Research Associate, from 2013 to 2014, and with
Delft Center for Systems and Control as a
PostDoc, from 2014 to 2017. Since 2018, he has been with the Beijing
Institute of Technology, Beijing, China, as a Full Professor. His research
interests include system identification, distributed optimization, and
optical imaging.
\end{IEEEbiography}
\begin{IEEEbiography}[{\includegraphics[width=0.95in,height=1.2in,clip]{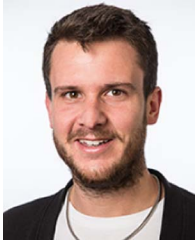}}]
{Filippo Fabiani} received the B.Sc. degree in
bio-engineering, the M.Sc. degree in automatic
control engineering, and the Ph.D. degree in
automatic control from the University of Pisa,
Pisa, Italy, in 2012, 2015, and 2019, respectively. He is currently an Assistant Professor with
the IMT School for Advanced Studies Lucca,
Lucca, Italy. In 2018–2019, he was Postdoctoral
Research Fellow with the Delft Center for Systems and Control, TU Delft, Delft, The Netherlands, while in 2019–2022, he was a Postdoctoral Research Assistant with the Control Group, Department of Engineering Science, University of Oxford, Oxford, U.K. His research interests include game theory, optimization and control of complex uncertain
systems, with applications in power networks and automated driving.
\end{IEEEbiography}

\end{document}